%% file: manuscript_report.tex
\documentclass[lettersize,journal,romanappendices]{IEEEtran}

\input{packages}
\input{commands}
\input{notation}

\begin{document}

\title{\huge\mbox{Vandermonde Trajectory Bounds for Linear Companion Systems}}

\author{\"Om\"ur Arslan and Aykut \.{I}\c{s}leyen 
\thanks{The authors are with the Department of Mechanical Engineering, Eindhoven University of Technology, P.O. Box 513, 5600 MB Eindhoven, The Netherlands. The authors are also affiliated with the Eindhoven AI Systems Institute. Emails:  \{o.arslan, a.isleyen\}@tue.nl}%
}

\markboth{Technical Report, February~2023}%
{Arslan \MakeLowercase{\textit{et al.}}: Vandermonde Trajectory Bounds for Linear Companion Systems}

\maketitle

\begin{abstract}
Fast and accurate safety assessment and collision checking are essential for motion planning and control of highly dynamic autonomous robotic systems.
Informative, intuitive, and explicit motion trajectory bounds enable explainable and time-critical safety verification of autonomous robot motion.
In this paper, we consider feedback linearization of nonlinear systems in the form of proportional-and-higher-order-derivative (PhD) control corresponding to companion dynamics. 
We introduce a novel analytic convex trajectory bound, called \emph{Vandermonde simplex}, for high-order companion systems, that is given by the convex hull of a finite weighted combination of system position, velocity, and other relevant higher-order state variables.
Our construction of Vandermonde simplexes is built based on expressing the solution trajectory of companion dynamics in a newly introduced family of \emph{Vandermonde basis functions} that offer new insights for understanding companion system motion compared to the classical exponential basis functions.
In numerical simulations, we demonstrate that Vandermonde simplexes offer significantly more accurate motion prediction (e.g., at least an order of magnitude improvement in estimated motion volume) for describing the motion trajectory of companion systems compared to the standard invariant Lyapunov ellipsoids as well as exponential simplexes built based on exponential basis functions.
\end{abstract}


\section{Introduction}
\label{sec.Introduction}

\IEEEPARstart{M}{otion} prediction plays a key role in safety assessment and constraint satisfaction of autonomous intelligent systems.
Informative and analytical trajectory bounds for motion prediction enable fast and accurate risk assessment and collision checking in motion planning and control of highly dynamic robotic systems around obstacles \cite{isleyen_vandewouw_arslan_RAL2022}.
Existing explicit motion prediction approaches for bounding closed-loop system motion mainly rely on invariant sets (e.g., of Lyapunov functions) under feedback control, where system motion is guaranteed to stay in a motion set associated with control \cite{blanchini_Automatica1999}.
However, such set invariance methods are known to be conservative, because motion prediction is not only valid for the immediate system state but also holds for many other system states contained in the same invariant motion set.  
In this paper, we introduce a novel analytic convex trajectory bound, \emph{Vandermonde simplex}, for linear companion systems --- a common dynamic feedback linearization model of nonlinear systems via embedding proportional-and-higher-order-derivative (PhD) control dynamics \cite{charlet_levine_marino_SCL1989, oriolo_deluca_vendittelli_TCST2002, chang_eun_TAC2017, mistler_benallegue_msirdi_ROMAN2001, dandreanovel_bastin_campion_ICRA1992, zhou_schwager_ICRA2014}.
In addition to its simple and intuitive construction, the proposed Vandermonde simplexes have a stronger direct dependency on the initial system state and control parameters compared to Lyapunov level sets and so can capture companion system motion more accurately, as illustrated in \reffig{fig.VandermondeSimplex}.
This improvement is mainly due to expressing the solution trajectory of companion dynamics in a new family of \emph{Vandermonde basis functions} that
offer new insights about companion system motion and a numerically more stable trajectory bound than the standard exponential basis functions.

\begin{figure}[t]
\centering
\begin{tabular}{@{}c@{}c@{}c@{}}
\includegraphics[width=0.333\columnwidth]{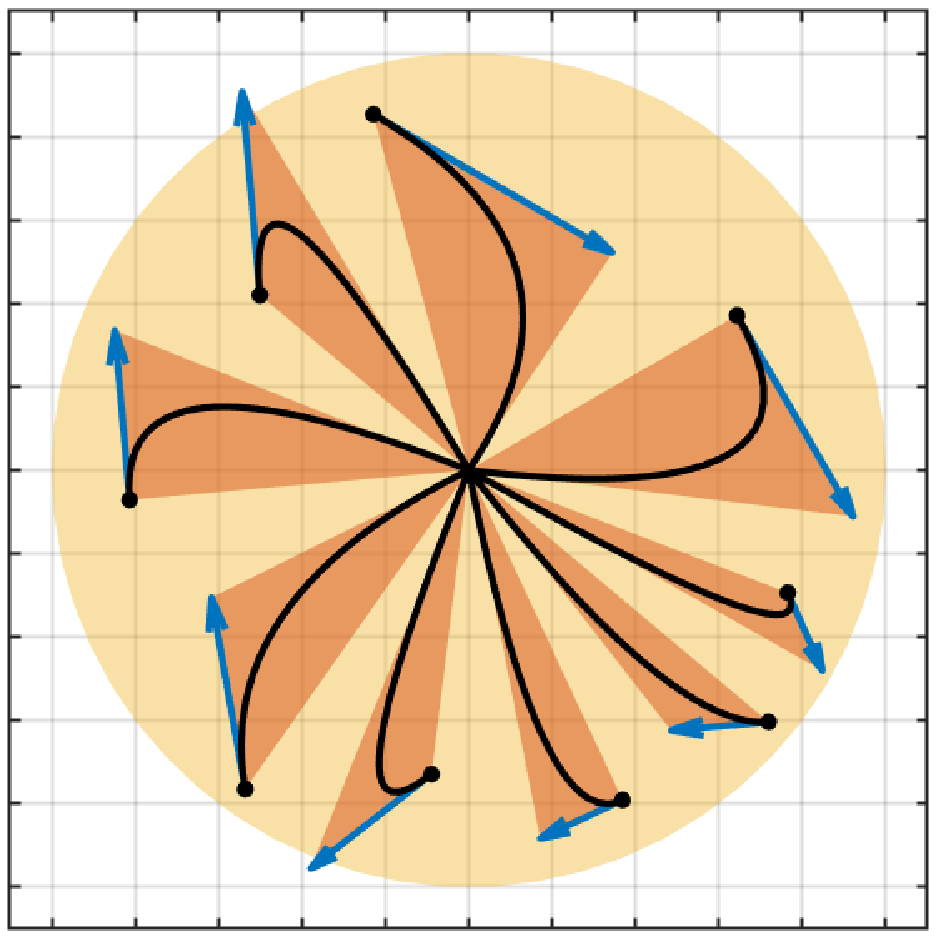} 
&
\includegraphics[width=0.333\columnwidth]{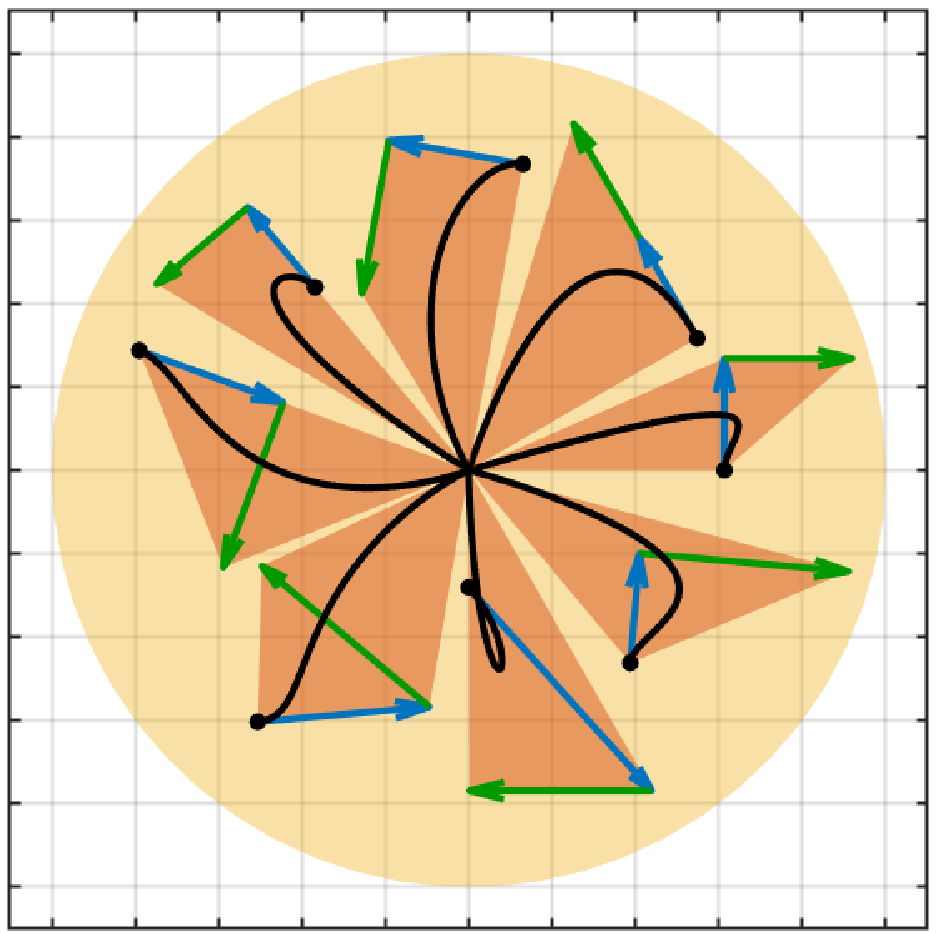} 
&
\includegraphics[width=0.333\columnwidth]{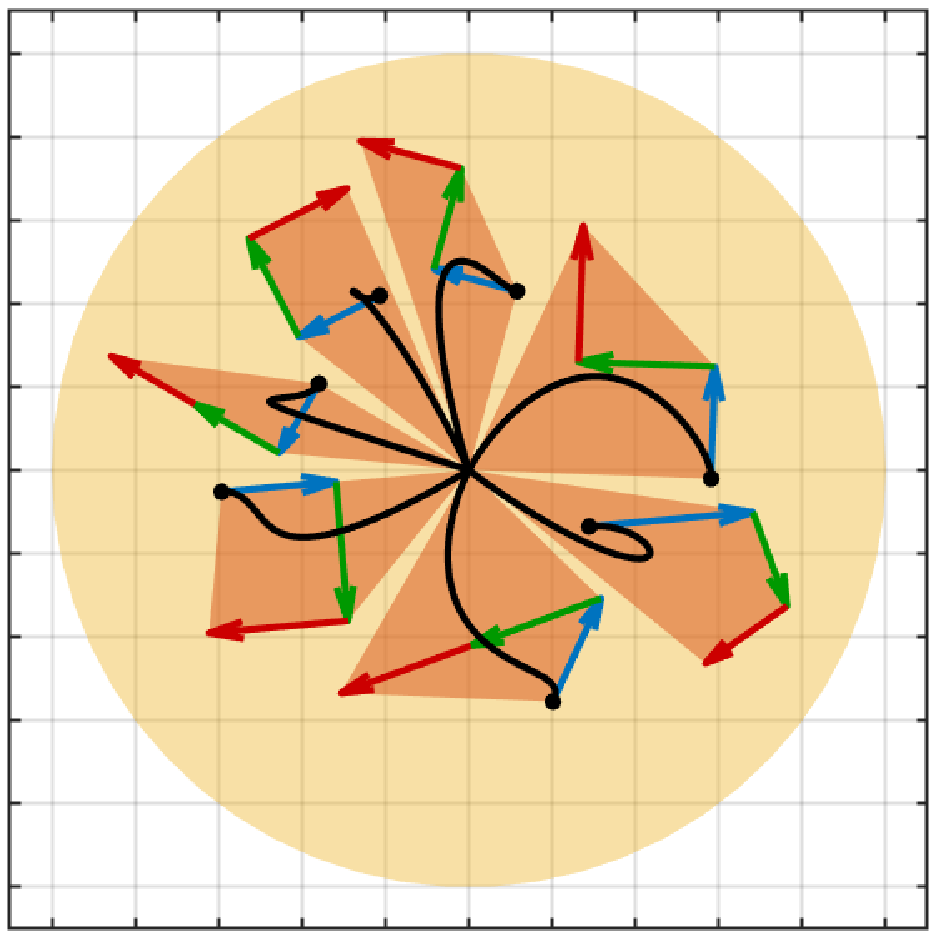} 
\end{tabular}
\vspace{-2mm}
\caption{Vandermonde simplexes (orange) for bounding the motion trajectory of the (left) second-order, (middle) third-order, (right) fourth-order companion systems starting from different initial states contained in the same projected Lyapunov ellipsoid (yellow). Here, the initial system position is indicated by a black dot and its initial scaled velocity (blue), acceleration (green), and jerk (red) are illustrated by arrows, where the scaling corresponds to Vandermonde simplex vertices. Vandermonde simplexes cover a smaller region and so more accurate compared to Lyapunov ellipsoids because Vandermonde simplexes have a stronger dependency on the initial system state and control parameters.}
\label{fig.VandermondeSimplex}
\vspace{-2mm}
\end{figure}

\subsection{Motivation and Related Literature}
\label{sec.Literature}

\subsubsection{Set Invariance in Control}
In addition to their essential use in stability analysis, invariant sets \cite{blanchini_Automatica1999} such as Lyapunov sublevel sets find significant application as a safety guard and verification tool in constrained control and optimization of dynamical systems to meet system constraints \cite{mayne_etal_Automatica2000, bemporad_casavola_mosca_TAC1997, prajna_jadbabaie_HSCC2004, kolmanovsky_garone_dicairano_ACC2014}.
A well-known limitation of invariant sets for safety verification and constraint checking is their conservativeness --- an invariant set, by definition, bounds the expected future system motion starting from any initial state within the set, but is not necessarily specific enough to the immediate system state that might be more relevant to system operation.
Simple shapes (e.g., ellipsoid and polytope) of invariant sets are often preferred in practice for computational efficiency, but this comes with a trade-off between representation complexity and conservativeness \cite{blanchini_Automatica1999}. 
To exploit this trade-off and mitigate conservativeness, invariant sets are applied together with a finite-horizon forward simulation of the system, for example, in model predictive control \cite{mayne_etal_Automatica2000} as a terminal condition to ensure stability and recursive feasibility \cite{lee_cannon_kouvaritakis_Automatica2005, lofberg_Automatica2012, kerrigan_maciejowski_CDC2000}. 
Similarly, existing approaches in controlled invariance based on control Lyapunov and barrier functions generally relax the invariance requirements away from the boundary of the constrain set to limit conservatism \cite{wieland_allgower_SNCS2007, ames_xu_grizzle_tabuada_TAC2017, ames_etal_ECC2019, gurriet_etal_CDC2018}.  
In this paper, by completely abandoning set invariance, we develop simple explicit convex trajectory bounds for companion systems that can be used as an accurate terminal condition or control barrier certificate in model predictive control with dynamic feedback linearization.

\subsubsection{Reachability Analysis}

A widely used computational approach for safety verification \cite{althoff_dolan_TR02014}, domain of attraction analysis \cite{limon_alamo_camacho_Automatica2005}, constrained control synthesis \cite{schurmann_etal_TAC2022}, and discrete abstraction \cite{lafferriere_pappas_yovine_JSC2001} of dynamical control systems is reachability analysis --- the estimation of a set of visited states of a dynamical system associated with some admissible sets of initial/goal states and control inputs over a finite or infinite horizon \cite{althoff_frehse_girard_ARCRAS2021}.
Determining reachable sets exactly is difficult because they often have complex and arbitrary shapes.
Existing reachability analysis methods mainly focus on computationally efficient inner (under) and outer (over) approximation of reachability sets by exploiting the structure of dynamical systems and set representation; for example, using Lyapunov-like invariant sets \cite{prajna_jadbabaie_HSCC2004}, the Hamilton-Jacobi equation \cite{mitchell_tomlin_HSCC2000}, system decomposition \cite{tiwari_HSCC2003}, semidefinite programming \cite{boyd_etal_LMI1994}, and simple geometric shapes (e.g., ellipsoids \cite{kurzhanski_varaiya_HSCC2000}, polytopes \cite{dang_maler_HSCC1998}, and separating hyperplanes  \cite{graettinger_krogh_JOTA1991}).
A common consensus on the use of invariant sets in reachability analysis is they improve computational efficiency but they can be significantly conservative. 
Another related concept to reachability analysis that avoids set invariance is output admissible sets \cite{gilbert_tan_TAC1991} --- an initial system state is output admissible if the resulting system output is contained in a constraint set for all future time.
Although maximal output admissible sets can be effectively computed as linear inequalities for stable and observable discrete-time linear systems, they are difficult to compute for continuous-time linear systems \cite{gilbert_tan_TAC1991}. 
The convex simplicial companion trajectory bounds presented in this paper might serve as a building block and tool for fast and accurate numerical approximation of reachable and output admissible sets of nonlinear systems via dynamic feedback linearization.

\subsubsection{Motion Prediction}

In robotics, motion prediction (i.e., anticipating the future motion of an autonomous agent) has recently received significant attention for  
safety assessment, control, and planning of autonomous robots in dynamic environments around people and other robots \cite{rudenko_etal_IJRR2020}. 
Most existing motion prediction algorithms use physical motion dynamics, characteristic motion patterns, high-level motion planning, and their combinations.
Physics-based motion prediction is often performed using the forward simulation of a simple dynamical system model constructed based on physics laws; for example, by assuming (piecewise) constant velocity/acceleration/turning-rate motion models \cite{schubert_richter_wanielik_ICIF2008}.
Pattern-based motion prediction uses pre-defined and learned motion patterns for modeling future system behavior \cite{schreier_willert_adamy_TITS2016, bennewitz_etal_IJRR2005}.
Planning-based motion prediction identifies a high-level motion objective (e.g., an intended goal position \cite{rehder_kloeden_ICCV2015, rehder_etal_ICRA2018} or a cost function \cite{wulfmeier_etal_IJRR2017}) from observed motion trajectories to combine with motion planning to predict the future system motion.
We believe that control-based feedback motion prediction built based on motion trajectory bounds under specific control offers a new perspective for closing the gap between motion prediction and motion control.  
We show in our recent studies that accurate feedback motion prediction (e.g., Vandermonde simplexes) is essential for safe and agile robot motion planning and control around obstacles \cite{isleyen_vandewouw_arslan_RAL2022, isleyen_vandewouw_arslan_arXiv2022, arslan_arXiv2022}.

\subsection{Contributions and Organization of the Paper}

This paper introduces novel analytic convex trajectory bounds for linear companion systems that can be used for fast and accurate safety assessment, motion planning and control of robotic systems via dynamic feedback linearization \cite{charlet_levine_marino_SCL1989, oriolo_deluca_vendittelli_TCST2002, chang_eun_TAC2017, mistler_benallegue_msirdi_ROMAN2001, dandreanovel_bastin_campion_ICRA1992, zhou_schwager_ICRA2014}. 
In summary, our main contributions are three folds:
\begin{enumerate}[i)] 
\item a new family of Vandermonde basis functions for representing the solution trajectory of companion dynamics that provides new insights about companion motion,
\item a new family of explicit convex  simplicial trajectory bounds, including Vandermonde and exponential simplexes, for linear companion systems, built based on the characteristic properties of companion basis functions,
\item a systematic comparison of Vandermonde simplexes with Lyapunov ellipsoids and exponential simplexes in numerical simulations to demonstrate their effectiveness and accuracy in capturing companion motion. 
\end{enumerate}
Our key technical result is a simple, interpretable, and accurate convex trajectory bound for companion systems, which we present below in \refthm{thm.VandermondeSimplex} and dedicate the rest of the paper to describing its construction and comparison with alternatives.
\begin{theorem} \label{thm.VandermondeSimplex}
\emph{\mbox{\!(Vandermonde Simplexes for Companion Systems)}} Consider an $\sorder^{\text{th}}$-order linear time-invariant dynamical system of the companion form
\begin{align}\label{eq.CompanionDynamicsExplicit}
\sposn^{(\sorder)} = -\sgain_{0, \slambdaset}\sposn^{(0)} - \sgain_{1, \slambdaset} \sposn^{(1)} - \ldots -\sgain_{\sorder-1, \slambdaset} \sposn^{(\sorder-1)},
\end{align}
where $\sposn^{(k)} := \frac{\diff^{k} }{\diff t^k} \sposn$ is the $k^{\text{th}}$ time derivative of the $\sdim$-dimen-sional system variable $\sposn \!\in\! \R^{\sdim}$, and $\sgain_{0, \slambdaset}, \ldots, \sgain_{\sorder-1, \slambdaset} \in \R$  are constant scalar control gains that ensure real negative characteristic polynomial roots $\slambdaset = (\slambda_1,  \ldots, \slambda_\sorder) \in \R_{<0}^{\sorder}$.

The system trajectory $\sposn(t)$, starting at $t = 0$ from any initial system state $\sstate_0 := (\sposn_0^{(0)},  \ldots, \sposn_0^{(\sorder-1)}) \!\in\! \R^{\sorder \times \sdim}$, is contained for all future times $t \geq 0$ in the \emph{Vandermonde simplex} $\vsplx_{\slambdaset} (\sstate_0)$ that is defined as\footnote{Summation over the empty set is assumed to be zero. Hence, the Vandermonde simplex in \refeq{eq.VandermondeSimplex} has the following form
\begin{align*}
\vsplx_{\slambdaset} ( \sstate_0) = \conv \plist{\!0, \sposn^{(0)}_0\!, \sposn^{(0)}_0\! + \!\frac{\sgain_{1, \slambdaset_{\neg \max}}}{\sgain_{0, \slambdaset_{\neg \max}}} \sposn^{(1)}_0\!, \ldots, \sum_{i=0}^{\sorder-1}\! \frac{\sgain_{i, \slambdaset_{\neg \max}}}{\sgain_{0,\slambdaset_{\neg \max}}} \sposn^{(i)}_0\!\!}\!. 
\end{align*} 
 } 
{
\begin{align}\label{eq.VandermondeSimplex}
\sposn(t) \!\in\! \vsplx_{\slambdaset}( \sstate_0)\! := \! \conv\plist{  \scalebox{0.95}{$\sum\limits_{j = 0}^{i-1}\dfrac{\sgain_{j, \slambdaset_{\neg \max}}}{\sgain_{0,\slambdaset_{\neg \max}}}$} \sposn^{(j)}_0 \!\bigg | \scalebox{0.95}{$i = 0, \ldots, \sorder$} \!}\!,\! \!\!
\end{align} 
}%
where $\conv$ denotes the convex hull operator,  $\slambdaset_{\neg \max}$ is the $(\sorder -1)$-element subset of $\slambdaset$ excluding  one maximal element that equals to $\max(\slambdaset)$, and  $\sgain_{0, \slambdaset_{\neg \max}}, \ldots, \sgain_{\sorder-2, \slambdaset_{\neg \max}}$ are the coefficients associated with $\slambdaset_{\neg \max}$, and $\sgain_{\sorder-1, \slambdaset_{\neg \max}} = 1$. 
\end{theorem}
\begin{proof}
See \refapp{app.VandermondeSimplex}.
\end{proof}

The rest of the paper is organized as follows. 
\refsec{sec.CompanionSystems} provides essential background on linear companion systems.
\refsec{sec.CompanionTrajectoryVandermondeBasis} presents Vandermonde basis functions and their important properties for understanding companion motion trajectory.
\refsec{sec.ConvexCompanionTrajectoryBounds} describes how to construct convex trajectory bounds for companion systems using Vandermonde and exponential basis functions as well as the classical Lyapunov theory.
\refsec{sec.NumericalSimulations} demonstrates the effectiveness of Vandermonde simplexes for capturing companion motion compared to exponential simplexes and Lyapunov ellipsoids in numerical simulations.
\refsec{sec.Conclusion} concludes with a summary of our contributions and future research directions.

\section{Linear Companion Systems}
\label{sec.CompanionSystems}

This section provides a brief background on linear companion systems and their trajectory representation using exponential basis functions and Vandermonde matrices. 

\subsection{Companion Dynamics}
\label{sec.CompanionDynamics}

\begin{definition} \label{def.CompanionSystem}
(\emph{Linear Companion Systems})
A linear \emph{companion system} is an $\sorder^{\text{th}}$-order fully-actuated dynamical system, with variable (e.g., position) $\sposn \in \R^{\sdim}$,
 that  evolves in the $\sdim$-dimensional Euclidean space $\R^{\sdim}$ under proportional-and-higher-order-derivative  (PhD) control as 
\begin{align} \label{eq.CompanionDynamics}
\sposn^{(\sorder)} & = - \sum_{k = 0}^{\sorder-1}  \sgain_{k,\slambdaset} \sposn^{(k)}     
\end{align} 
where $\sposn^{(k)} = \frac{\mathrm{d}^{k} }{\mathrm{d}t^k} \sposn$  and $\sgain_{0,\slambdaset}, \ldots, \sgain_{\sorder-1,\slambdaset} \in \R$ are fixed real control gains (a.k.a. \emph{companion coefficients}) that result in the complex characteristic polynomial roots $\slambdaset = (\slambda_1, \ldots, \slambda_{\sorder}) \in \C^{\sorder}$ for the companion dynamics in \refeq{eq.CompanionDynamics}, i.e.,
\begin{align}\label{eq.CharacteristicPolynomial}
\prod_{k=1}^{\sorder} (\slambda - \slambda_i) =  \slambda^\sorder + \sum_{k = 0}^{\sorder-1}  \sgain_{k,\slambdaset} \slambda^k.
\end{align}
\end{definition}

\bigskip

Note that  the companion coefficient $\sgain_{k,\slambdaset}$ can be explicitly determined using the characteristic polynomial roots $\slambdaset$  as%
\footnote{Multiplication over the empty set is assumed to be one. Hence, the control gain $\sgain_{k, \slambdaset}$ in \refeq{eq.ControlGain} is also well defined for $k = \sorder$ and yields $\sgain_{n, \slambdaset} =1$.}\,%
\footnote{Determining control gains $\sgain_{0, \slambdaset}, \ldots, \sgain_{\sorder-1, \slambdaset}$ from characteristic polynomial roots $\slambda_1, \ldots, \slambda_\sorder$ is numerically more stable than vice versa. In MATLAB, one can easily perform these conversions as
\begin{align*}
[1, \sgain_{\sorder-1, \slambdaset}, \ldots, \sgain_{0, \slambdaset}] &= \mathtt{poly}([\slambda_1, \ldots, \slambda_\sorder]), 
\\
[\slambda_1, \ldots, \slambda_\sorder] &= \mathtt{roots}([1, \sgain_{\sorder-1, \slambdaset}, \ldots, \sgain_{0, \slambdaset}]).
\end{align*}  
}
\begin{align}\label{eq.ControlGain}
\sgain_{k,\slambdaset} = (-1)^{ \sorder - k}\sum_{\substack{I \subseteq \{1, \ldots, \sorder\} \\ |I| = \sorder - k} } \prod_{i \in I} \slambda_i, 
\end{align}    
where $k \in \clist{0, \ldots, \sorder}$ and $|.|$ denotes the number of elements of a set.
It is also convenient to have  $\sgain_{\sorder, \slambdaset} = 1$ and  $\sgain_{k, \slambdaset} = 0$ for any $k \not \in [0, \sorder]$  for the recursive use of companion coefficients (see \reflem{lem.CompanionCoefficientRecursion}).

Using state-space representation, one can alternatively represent the $\sorder^{\text{th}}$-order companion system in \refeq{eq.CompanionDynamics} as a first-order higher-dimensional dynamical system as \cite{chen_LinearSystemTheory1998}
\begin{align} \label{eq.StateSpaceCompanionDynamics}
\dot{\sstate} = \cmat_{\slambdaset} \sstate
\end{align}
where $\sstate \in \R^{\sorder \times \sdim}$ and $\cmat_{\slambdaset} \in \R^{\sorder \times \sorder}$, respectively, denote the system state and the companion matrix that are defined as
\begin{align}\label{eq.CompanionStateMatrix}
\sstate = \!
\scalebox{0.85}{$\begin{bmatrix}
\sposn^{(0)} \\
\sposn^{(1)} \\
\sposn^{(2)}\\
\vdots \\
\sposn^{(\sorder-1)}
\end{bmatrix}$} , 
\,\, 
\cmat_{\slambdaset} = \! 
\scalebox{0.85}{$
\left[
\begin{array}{@{}c@{\hspace{2mm}}c@{\hspace{2mm}}c@{\hspace{2mm}}c@{\hspace{2mm}}c@{}} 
0 & 1 & 0 &...& 0 \\
0 & 0 & 1 &...& 0 \\
\vdots&\vdots& \vdots &\ddots&\vdots \\
0 & 0 &  0 & \dots  & 1 \\
-\sgain_{0,\slambdaset}&-\sgain_{1,\slambdaset} & -\sgain_{2,\slambdaset} & \dots&-\sgain_{\sorder -1, \slambdaset} 
\end{array}
\right]
$}.\!\!\!
\end{align}
Note that  the characteristic polynomial of $\cmat_{\slambdaset}$ is given by \refeq{eq.CharacteristicPolynomial} and so the eigenvalues  of $\cmat_{\slambdaset}$ are $\slambdaset = (\slambda_1, \ldots, \slambda_\sorder)$.
Also, by abusing the notation, the companion state $\sstate$ is sometimes represented as a $\sorder\sdim \times 1$ vector in $\R^{\sorder\sdim \times 1}$ instead of a $\sorder \times \sdim$ matrix in $\R^{\sorder \times \sdim}$, which shall be clear from the content.

\subsection{Companion Trajectory}
\label{sec.CompanionTrajectory}

In linear system theory \cite{chen_LinearSystemTheory1998}, it is well known that the solution of the state-space companion dynamics in \refeq{eq.StateSpaceCompanionDynamics}, starting at $t = 0$ from any initial state $\sstate_0=(\sposn_{0}^{(0)}, \ldots, \sposn_{0}^{(\sorder-1)}) \in \R^{\sorder \times \sdim}$, can be determined using matrix exponent as 
\begin{align}
\sstate(t) = e^{\cmat_{\slambdaset} t} \sstate_0, \quad   \quad \forall t \geq 0.
\end{align}
If the companion matrix $\cmat_{\slambdaset}$ has distinct eigenvalues (i.e., $\slambda_i \neq \slambda_j$ for all $i \neq j$), then  it is explicitly diagonalizable as
\begin{align}\label{eq.CompanionEigenDecomposition}
\cmat_{\slambdaset} = \vmat_{\slambdaset} \diag(\slambdaset) \vmat_{\slambdaset}^{-1} 
\end{align}
where $\diag(\slambdaset)$ denotes the diagonal matrix with vector $\slambdaset$ on the diagonal, and $\vmat_{\slambdaset}$ is the Vandermonde matrix defined as%
\begin{align} \label{eq.VandermondeMatrix}
\vmat_{\slambdaset} &:= \begin{bmatrix} 
1 & 1 & \dots & 1 \\
\slambda_1 & \slambda_2 & \dots & \slambda_\sorder \\ 
\slambda_1^{2} & \slambda_2^{2}& \dots & \slambda_\sorder^{2} \\
\vdots & \vdots & \ddots & \vdots \\ 
\lambda_1^{\sorder-1} & \lambda_2^{\sorder-1} & \dots & \lambda_\sorder^{\sorder-1} 
\end{bmatrix}
\end{align}
whose inverse is explicitly given by\cite{neagoe_SPL1996}%
\footnote{%
The elements of $\vmat_{\slambdaset}$ and  $\vmat_{\slambdaset}^{-1}$ on the $i^{\text{th}}$ row and $j^{\text{th}}$ column, where $i,j \in \{1, \ldots, \sorder\}$,  are given by
\begin{align*}
\blist{\vmat_{\slambdaset}\big.}_{ij} &= \slambda_{j}^{i-1}, 
\quad \text{and} \quad  
\blist{\big.\vmat_{\slambdaset}^{-1}}_{ij} = (-1)^{\sorder-1}\frac{\sgain_{j-1, \slambdaset_{\neg i}}}{\prod\limits_{k\neq i}(\slambda_k - \slambda_i)}. 
\end{align*}
In  \cite{neagoe_SPL1996}, the inverse of Vandermonde matrices is expressed slightly differently as \mbox{$\blist{\vmat_{\slambdaset}^{-1}}_{ij} \!\! \!=\! (-1)^{j-1}\!\frac{\sigma_{\sorder - j, \slambdaset_{\neg i}} }{\!\prod\limits_{k\neq i}\!(\slambda_k - \slambda_i)}$} in terms of
\begin{align*}
\sigma_{k,\slambdaset} := \sum_{\substack{I \subseteq \{1, \ldots, \sorder\}\\ |I| = k}} \prod_{i \in I} \slambda_i \nonumber
\end{align*}
which is related to the companion coefficients $\sgain_{k, \slambdaset}$ in \refeq{eq.ControlGain} as
\begin{align*}
\sigma_{k, \slambdaset} =  (-1)^{k} \sgain_{\sorder - k, \slambdaset}, \quad \text{and} \quad 
\sgain_{k, \slambdaset} =  (-1)^{\sorder - k} \sigma_{\sorder - k, \slambdaset}.
\end{align*}
} 
\begin{align} \label{eq.VandermondeInverse}
\vmat_{\slambdaset}^{-1} &= 
\scalebox{0.9}{$
(-1)^{\sorder-1} \!
\blist{
 \begin{array}{@{}c@{\hspace{1mm}}c@{\hspace{1mm}}c@{\hspace{1mm}}c@{}}
\frac{\sgain_{0, \slambdaset_{\neg 1}}}{\prod\limits_{k\neq 1}(\slambda_k - \slambda_1)} & \frac{\sgain_{1, \slambdaset_{\neg 1}}}{\prod\limits_{k\neq 1}(\slambda_k - \slambda_1)} & \ldots & \frac{\sgain_{\sorder-1, \slambdaset_{\neg 1}}}{\prod\limits_{k\neq 1}(\slambda_k - \slambda_1)}  
\\
\frac{\sgain_{0, \slambdaset_{\neg 2}}}{\prod\limits_{k\neq 2}(\slambda_k - \slambda_2)} & \frac{\sgain_{1, \slambdaset_{\neg 2}}}{\prod\limits_{k\neq 2}(\slambda_k - \slambda_2)} & \ldots & \frac{\sgain_{\sorder-1, \slambdaset_{\neg 2}}}{\prod\limits_{k\neq 2}(\slambda_k - \slambda_2)}\\
\vdots & \vdots  & \ddots & \vdots \\
\frac{\sgain_{0, \slambdaset_{\neg \sorder}}}{\prod\limits_{k\neq \sorder}(\slambda_k - \slambda_\sorder)} & \frac{\sgain_{1, \slambdaset_{\neg \sorder}}}{\prod\limits_{k\neq \sorder}(\slambda_k - \slambda_\sorder)} & \ldots & \frac{\sgain_{\sorder-1, \slambdaset_{\neg \sorder}}}{\prod\limits_{k\neq \sorder}(\slambda_k - \slambda_\sorder)}   
\end{array}
}
$}
\!,
\!\!\!
\end{align}
where $\slambdaset_{\neg i} := \plist{\slambda_1, \ldots, \slambda_{i-1}, \slambda_{i+1}, \ldots, \slambda_\sorder}$ is obtained by removing  $i^{\text{th}}$ element from $\slambdaset$, and the companion coefficients $\sgain_{k, \slambdaset}$ are defined as in \refeq{eq.ControlGain}.  
Hence, for any pairwise distinct characteristic polynomial roots $\slambdaset=\plist{\slambda_1, \ldots, \slambda_\sorder} \in \C^{\sorder}$, the  solution of the  state-space companion dynamics in \refeq{eq.StateSpaceCompanionDynamics} can be obtained  using the eigendecomposition of $\cmat_{\slambdaset}$ in \refeq{eq.CompanionEigenDecomposition} as 
\begin{align}
\sstate(t) = e^{\cmat t} \sstate_0 = \vmat_{\slambdaset} \diag(e^{\slambdaset t}) \vmat_{\slambdaset}^{-1} \sstate_0 
\end{align} 
where $e^{\slambdaset t} = [e^{\slambda_1 t}, \ldots, e^{\slambda_{\sorder} t}]$. 
Therefore, the system trajectory $\sposn(t)$ that solves the $\sorder^{\text{th}}$-order companion dynamics in \refeq{eq.CompanionDynamics} is given by
\begin{subequations} \label{eq.SystemTrajectoryExponentialBasis}
\begin{align}
\sposn(t) & =  \begin{bmatrix}
1 & 0 &\dots & 0
\end{bmatrix} \sstate(t) 
\\
&=   \begin{bmatrix}
1 & 0 &\dots & 0
\end{bmatrix} \vmat_{\slambdaset} \diag(e^{\slambdaset t}) \vmat_{\slambdaset}^{-1} \sstate_0 
\\
& = \begin{bmatrix}
1 & 1 &\dots & 1
\end{bmatrix} \diag(e^{\slambdaset t}) \vmat_{\slambdaset}^{-1} \sstate_0, 
\\
& = e^{\slambdaset t}  \vmat_{\slambdaset}^{-1} \sstate_0,
\end{align}
\end{subequations}
which is often  written using exponential basis functions as
\begin{align}\label{eq.TrajectoryExponentialBasis}
\sposn(t) &= \sum_{k=1}^{\sorder} \scoef_{k, \slambdaset}(\sstate_0) e^{\slambda_k t} 
\end{align}
using exponential trajectory coefficients $\scoef_{k, \slambdaset}(\sstate_0) \in \R^{\sdim}$ that corresponds to the $k^{\text{th}}$ row of $ \vmat_{\slambdaset}^{-1} \sstate_0$ and can be obtained using \refeq{eq.VandermondeInverse} as    
\begin{align}\label{eq.ExponentialTrajectoryCoefficients}
\scoef_{k, \slambdaset}(\sstate_0) = \frac{(-1)^{\sorder -1}}{\prod_{i \neq k}(\slambda_i - \slambda_k)}\sum_{i=1}^{\sorder-1} \sgain_{i,\slambdaset_{\neg k}} \sposn_0^{(k)}.
\end{align}
Even though it has a simple explicit form, the companion trajectory $\sposn(t)$ expressed in exponential basis in \refeq{eq.TrajectoryExponentialBasis} becomes numerical instability if the difference between any pair of eigenvalues in $\slambdaset$ is very small, which is mainly due to the singularity of exponential trajectory coefficients $\scoef_{k, \slambdaset}(\sstate_0) \in \R^{\sdim}$ in \refeq{eq.ExponentialTrajectoryCoefficients}.
In the following section, as an alternative to exponential basis, we introduce a new family of basis functions, called \emph{Vandermonde basis}, for expressing the companion motion trajectory that transfers such numerical stability issues from trajectory coefficients to basis functions, which becomes useful for constructing companion trajectory bounds in \refsec{sec.ConvexCompanionTrajectoryBounds}.

\section{Linear Companion System Trajectory \\ via Vandermonde Basis}
\label{sec.CompanionTrajectoryVandermondeBasis}

In this section, we introduce a new family of Vandermonde basis functions for expressing companion motion trajectories and present their important (nonnegativity, relative ordering, and boundedness) properties that are essential for understanding and bounding companion motion.  

\subsection{Vandermonde Basis Functions}
\label{sec.VandermondeBasisFunctions}

\begin{figure*}
\centering
\begin{tabular}{@{}ccc@{}}
\includegraphics[width=0.31\textwidth]{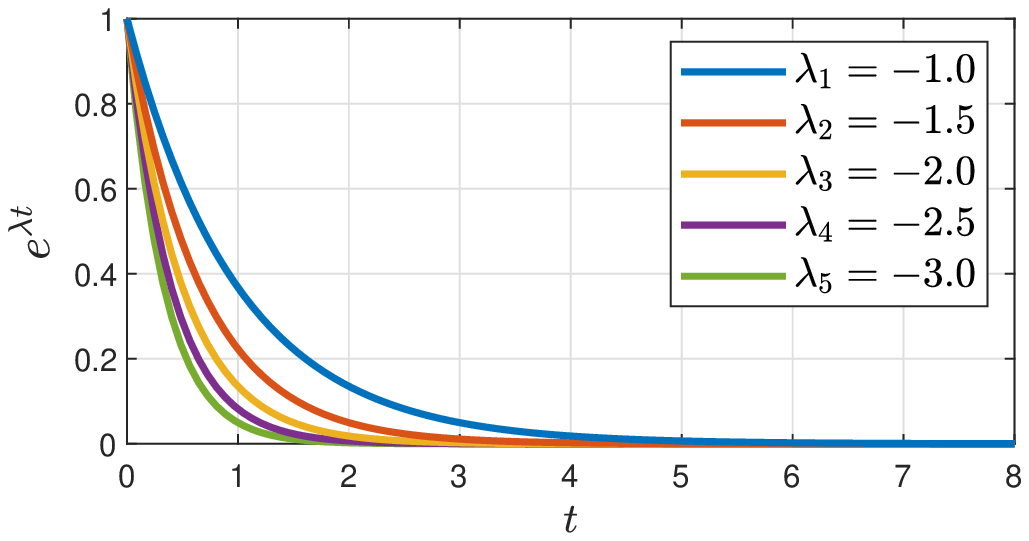} 
&
\includegraphics[width=0.31\textwidth]{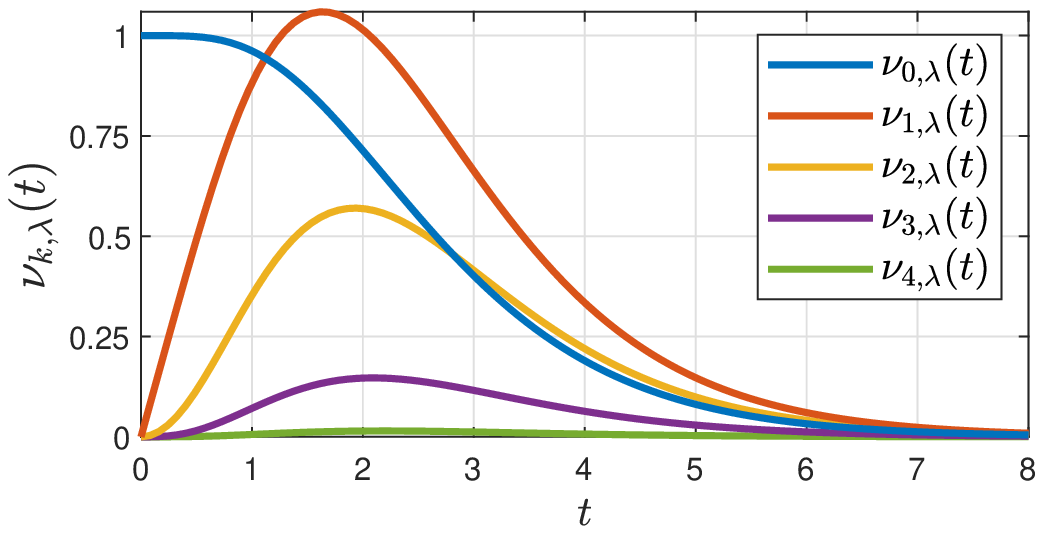} 
& 
\includegraphics[width=0.31\textwidth]{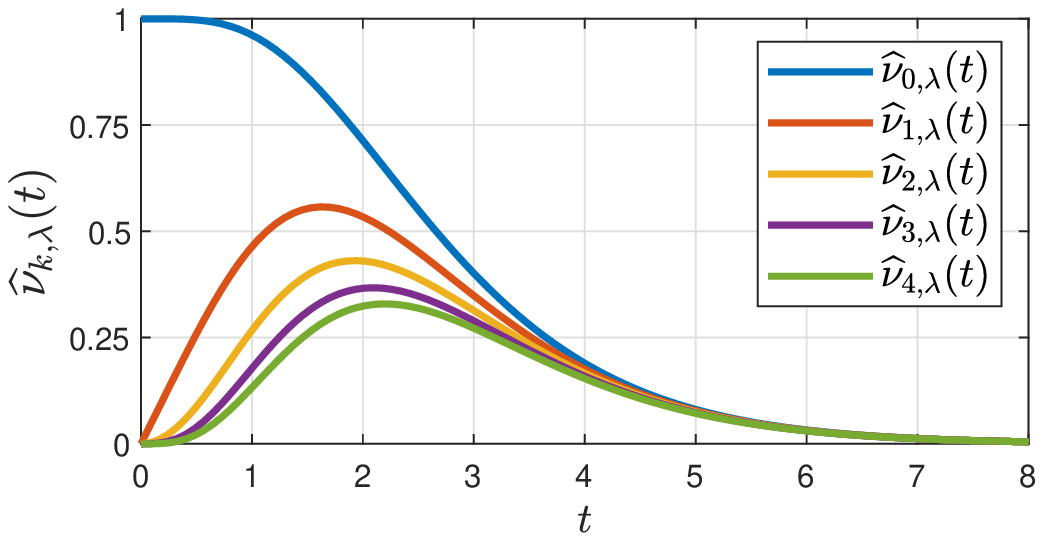} 
\end{tabular}
\vspace{-2mm}
\caption{Exponent (left) and  Vandermonde (middle) basis functions for linear companion systems, where $\slambdaset =\plist{-1.0, -1.5, -2.0, -2.5, -3.0}$. (right) Scaled Vandermonde  basis functions $\widehat{\vfunc}_{k, \slambdaset}(t):= \frac{\sgain_{0, \slambdaset_{\neg \max}}}{\sgain_{k, \slambdaset_{\neg \max}}}\vfunc_{k, \slambdaset}(t)$ demonstrate their ordering relation, i.e., $\widehat{\vfunc}_{i, \slambdaset}(t) \geq \widehat{\vfunc}_{j, \slambdaset}(t) $ for any $ 0 \leq i \leq j \leq |\slambdaset|-1$. }
\label{fig.ExponentialVandermondeBasisFunctions}
\vspace{-2mm}
\end{figure*}

\begin{definition} \label{def.VandermondeBasisFunctions}
(\emph{Vandermonde Basis})
For any complex vector $\slambdaset=(\slambda_1, \ldots, \slambda_\sorder) \in \C^\sorder$ with distinct elements, (i.e., $\slambda_k \neq \slambda_l$ for all $k \neq l$), the \emph{Vandermonde basis functions}, denoted by
\begin{align}\label{eq.VandermondeBasisVector}
\vvect_{\slambdaset}(t) = [\vfunc_{0, \slambdaset}(t), \ldots, \vfunc_{\sorder-1, \slambdaset}(t)] \in \C^{\sorder}
\end{align}
are defined as a linear transformation of exponential basis functions, illustrated in \reffig{fig.ExponentialVandermondeBasisFunctions}, 
\begin{align}\label{eq.ExponentialBasisVector}
e^{\slambdaset t} = [e^{\slambda_1 t}, \ldots, e^{\slambda_\sorder t}] \in \C^{\sorder}
\end{align}
via the inverse Vandermonde matrix $\vmat_{\slambdaset}^{-1}$ as 
\begin{align} \label{eq.Exponential2Vandermonde}
\vvect_{\slambdaset}(t) = e^{\slambdaset t} \vmat_{\slambdaset}^{-1}.  
\end{align}   
Hence, the change of basis from Vandermonde functions to exponentials is given by the Vandermonde matrix $\vmat_{\slambdaset}$ as
\begin{align}\label{eq.Vandermonde2Exponential}
e^{\slambdaset t}  = \vvect_{\slambdaset}(t) \vmat_{\slambdaset}.
\end{align}
\end{definition}

It follows from the explicit form of the Vandermonde matrix in \refeq{eq.VandermondeMatrix} and its inverse in \refeq{eq.VandermondeInverse} that Vandermonde basis functions can be expressed as a weighted combination of exponential basis functions  (and vice versa) as
\begin{align}
\vfunc_{k, \slambdaset}(t) &= (-1)^{\sorder -1} \sum_{i =1}^{\sorder} \frac{\sgain_{k, \slambdaset_{\neg i}}}{\prod_{j \neq i}(\slambda_j - \slambda_i)} e^{\slambda_i t} \label{eq.VandermondeSumOfExponential}
\\
e^{\slambda_k t}  &= \sum_{i = 0}^{\sorder -1} \slambda_k^i \vfunc_{i, \slambdaset}(t) \label{eq.ExponentialSumOfVandermonde}
\end{align}
where $\sgain_{k, \slambdaset_{\neg i}}$ is the companion coefficient defined in \refeq{eq.ControlGain} associated with $\slambdaset_{\neg i} = \plist{\slambda_1, \ldots, \slambda_{i-1}, \slambda_{i+1}, \ldots, \slambda_{\sorder}}$.
We also find it useful to define $\vfunc_{k,\slambdaset}(t) = 0$ for $k < 0$ and $k \geq \sorder$ for the recursive use of Vandermonde basis functions later.

Accordingly, one can expresses the trajectory of $\sorder^{th}$-order linear companion systems  in \refeq{eq.CompanionDynamics} using Vandermonde basis as: 

\begin{proposition}\label{prop.CompanionTrajectoryVandermondeBasis}
\emph{(Companion Trajectory in Vandermonde Basis)}
For an $\sorder^{th}$-order linear companion system in \refeq{eq.CompanionDynamics} with distinct characteristic polynomial roots $\slambdaset=(\slambda_1, \ldots, \slambda_{\sorder}) \in \C^{\sorder}$ with $\slambda_k \leq \slambda_l$ for all $k \neq l$, the system trajectory  $\sposn(t)$, starting  at $t=0$ from any initial state $\sstate_0 = (\sposn_0^{(0)}, \ldots,\sposn_0^{(\sorder-1)}) \in \R^{\sorder \times \sdim}$, is given in terms of Vandermonde basis functions $\vfunc_{\slambdaset}(t)$  by 
\begin{align} \label{eq.CompanionTrajectoryVandermondeBasis}
\sposn(t)= \vvect_{\slambdaset}(t) \sstate_0 = \sum_{k=0}^{\sorder-1} \vfunc_{k, \slambdaset}(t) \sposn_0^{(k)}  \quad \forall t \geq 0.
\end{align}
\end{proposition}
\begin{proof}
The result is a direct consequence of the classical trajectory solution expressed in exponential basis in \refeq{eq.SystemTrajectoryExponentialBasis} and the change of basis from exponentials to Vandermonde functions in  \refeq{eq.Exponential2Vandermonde}, i.e., $\sposn(t) = e^{\slambdaset t}  \vmat_{\slambdaset}^{-1}\sstate_{0} =  \vect{v}_{\slambdaset} (t) \sstate_{0}$.
\end{proof}

In practice, one often prefers to avoid underdamped second-order companion dynamics to prevent oscillatory system motion. 
Hence, using the companion motion trajectory in \refeq{eq.CompanionTrajectoryVandermondeBasis} expressed in Vandermonde basis, the notion of non-underdamped second-order companion systems can be intuitively extended to higher-order companion systems as non-overshooting.

\begin{definition}\label{def.NonovershootingCompanionSystem}
(\emph{Nonovershooting Companion Systems}) 
An $\sorder^{th}$-order linear companion system with distinct characteristic polynomial roots $\slambdaset=(\slambda_1, \ldots, \slambda_{\sorder}) \in \C^{\sorder}$ with $\slambda_k \neq \slambda_l$ for all $k \neq l$ is said to be \emph{nonovershooting} if the associated Vandermonde basis functions are nonnegative, i.e.,
\begin{align}
\vfunc_{k,\slambdaset}(t) \geq 0,  \quad \forall \, t \geq 0, \,  k = 0, \ldots, \sorder-1. 
\end{align}
\end{definition}

As one might expect from the second-order system case, a companion system is nonovershoting if its characteristic polynomial roots are real and negative (see \refprop{prop.VandermondeBasisNonnegativity}).

\subsection{Vandermonde Basis Properties}
\label{sec.VandermondeBasisProperties}

As an alternative to \refeq{eq.VandermondeSumOfExponential}, Vandermonde basis functions can be determined using \refeq{eq.Vandermonde2Exponential} and Cramer's rule so that it is easy to observe $\vfunc_{0,\slambdaset}(0) = 1$ and $\vfunc_{k, \slambdaset}(0) = 0$ for all $k \neq 0$. 

\begin{proposition}\label{prop.VandermondeBasisCramerRule}
\emph{(Cramer's Rule of Vandermonde Basis)}
For any $\slambdaset=(\slambda_1, \ldots, \slambda_\sorder) \in \C^\sorder$ with  $\slambda_k \neq \slambda_l$ for all $k \neq l$, each Vandermonde basis function is given by
\begin{align}\label{eq.VandermondeBasisCramerRule}
\quad \vfunc_{k, \slambdaset}(t) = \frac{\det \vmat_{k, \slambdaset}(t)}{\det \vmat_{\slambdaset}} \quad  \quad \forall \, k = 0, \ldots, (\sorder-1),
\end{align}
where  the row-exponential Vandermonde matrix $\vmat_{k, \slambdaset}(t)$ is obtained by replacing the $(k+1)^{th}$ row of the standard Vandermonde matrix $\vmat_{\slambdaset}$ with the exponential basis $e^{\slambdaset t}$ as~%
\footnote{The $i^{\text{th}}$-row and $j^{\text{th}}$-column element of $\vmat_{k, \slambdaset}(t)$ is given by 
\begin{align*}
\blist{\vmat_{k,\slambdaset}(t)\Big.}_{ij} & = 
\left \{ 
\begin{array}{cl}
e^{\slambda_j t}, & \mathrm{if} \,\,   i = k+1, \\
\lambda_j^{i-1}, & \mathrm{otherwise.}
\end{array} 
\right. 
\end{align*}
}
\begin{align}\label{eq.ExponentialVandermondeMatrix}
\vmat_{k, \slambdaset}(t) &= 
\begin{bmatrix} 
1 & 1 & \ldots & 1 \\
\slambda_1 & \slambda_2 & \ldots & \slambda_{\sorder} \\
\vdots & \vdots & \ddots & \vdots \\
\slambda_1^{k-1} & \slambda_2^{k-1} & \ldots & \slambda_{\sorder}^{k-1} \\
e^{\slambda_1 t} & e^{\slambda_2 t} & \ldots & e^{\slambda_\sorder t} \\ 
\slambda_1^{k+1} & \slambda_2^{k+1} & \ldots & \slambda_{\sorder}^{k+1} \\
\vdots & \vdots & \ddots & \vdots \\
\slambda_1^{\sorder-1} & \slambda_2^{\sorder-1} & \ldots & \slambda_{\sorder}^{\sorder-1}   
\end{bmatrix}  
\end{align}   
\end{proposition}
\begin{proof}
The result follows from the linear basis transformation \refeq{eq.Vandermonde2Exponential} and Cramer's rule for solving linear equations.
\end{proof}

As opposed to the exponential basis vector $e^{\slambdaset t}$ in \refeq{eq.ExponentialBasisVector}, the Vandermonde basis vector $\vvect_{\slambdaset}(t)$ in \refeq{eq.VandermondeBasisVector} does not depend on the order of elements of $\slambdaset = (\slambda_1, \ldots, \slambda_\sorder)$.

\begin{proposition}\label{prop.VandermondeBasisOrderIndependence}
\emph{(Root-Order Invariance of Vandermonde Basis)}
For any complex vector $\slambdaset \in \C^{\sdim}$  with distinct elements (i.e. $\slambda_k \neq \slambda_l$ for all $k \neq l$), the Vandermonde basis vector $\vvect_{\slambdaset}(t)$ is invariant under any permutation (i.e., rearrangement) of elements of $\slambdaset$ into $\widehat{\slambdaset}$, i.e., $\vvect_{\slambdaset}(t) = \vvect_{\widehat{\slambdaset}}(t)$.
\end{proposition}
\begin{proof}
By definition, both companion coefficients $\sgain_{k, \slambdaset}$ in \refeq{eq.ControlGain} and Vandermonde basis functions $\vfunc_{k, \slambdaset}(t)$ in \refeq{eq.VandermondeSumOfExponential} are independent of the order of elements of $\slambdaset$.
\end{proof}

Since Vandermonde matrices are the eigenbasis of companion matrices as seen in \refeq{eq.CompanionEigenDecomposition}, Vandermonde basis functions obey companion dynamics.
\begin{proposition}\label{prop.VandermondeDynamics}
\emph{(Vandermonde Basis Dynamics)} For any $\slambdaset \in \C^{\sorder}$ with pairwise distinct elements (i.e., $\slambda_k \neq \slambda_l$ for all $k \neq l$), the Vandermonde basis vector $\vvect_{\slambdaset}(t) \in \R^{1 \times \sorder}$ satisfies the (transposed) companion dynamics in \refeq{eq.CompanionDynamics} as
\begin{align}\label{eq.VandermondeDynamics}
\dot{\vvect}_{\slambdaset}(t) = \vvect_{\slambdaset}(t) \cmat_{\slambdaset}  \quad \forall t \geq 0
\end{align}
with the initial condition $\vvect_{\slambdaset}(0) = [1, 0, \ldots, 0]$, where $\cmat_{\slambdaset}$ is the associated companion matrix defined in \refeq{eq.CompanionStateMatrix}. 
\end{proposition}
\begin{proof}
By differentiating $\vvect_{\slambdaset}(t) = e^{\slambdaset t} \vmat_{\slambdaset}^{-1}$, one can verify that
\begin{align}
\dot{\vvect}_{\slambdaset}(t) &=   e^{\slambdaset t} \diag(\slambdaset) \vmat_{\slambdaset}^{-1}  \\
&=\vvect_{\slambdaset}(t) \vmat_{\slambdaset} \diag(\slambdaset) \vmat_{\slambdaset}^{-1} = \vvect_{\slambdaset}(t) \cmat_{\slambdaset}
\end{align}
where $\frac{\text{d}}{\text{d}t}e^{\slambdaset t} = e^{\slambdaset t}\diag(\slambdaset)$ and $e^{\slambdaset t} = \vvect_{\slambdaset}(t) \vmat_{\slambdaset}$.

Moreover, one can observe from \refeq{eq.ExponentialVandermondeMatrix} that $\vmat_{0,\slambdaset}(0) = \vmat_{\slambdaset}$, and $\vmat_{k,\slambdaset}(0)$ has two identical rows of all ones for $k \neq 0$. 
Hence, using \refeq{eq.VandermondeBasisCramerRule}, one conclude that $\vvect_{\slambdaset}(0) = [1, 0, \ldots, 0]$, which completes the proof.
\end{proof}

According to the Vandermonde basis dynamics in \refeq{eq.VandermondeDynamics}, each Vandermonde basis function $\vfunc_{k, \slambdaset}(t)$ evolves over time as
\begin{align}
\dot{\vfunc}_{k, \slambdaset}(t) = \vfunc_{k-1, \slambdaset}(t) - \sgain_{k,\slambdaset} \vfunc_{\sorder-1, \slambdaset}(t)
\end{align} 
for $k = 0, \ldots, \sorder-1$, where  $\vfunc_{-1, \slambdaset}(t)=0$ and $\sgain_{k,\slambdaset}$ is the companion coefficient defined as in \refeq{eq.ControlGain}.
To effectively handle the coupling between Vandermonde basis functions, we find it useful to describe Vandermonde basis dynamics recursively. 

\begin{proposition}\label{prop.VandermondeDynamicsRecursion}
\emph{(Recursive Vandermonde Basis Dynamics)} 
For any $\slambdaset = \plist{\slambda_1, \ldots, \slambda_\sorder} \in \C^{\sorder}$ with $\slambda_i \neq \slambda_j$ for all $i \neq j$, the Vandermonde basis functions satisfy for any $k= 0, \ldots, \sorder-1$
\begin{align}\label{eq.VandermondeDynamicsRecursion}
\dot{\vfunc}_{k, \slambdaset} (t) = \slambda_l \vfunc_{k, \slambdaset}(t) - \slambda_l \vfunc_{k, \slambdaset_{\neg l}} (t) +  \vfunc_{k-1, \slambdaset_{\neg l}} (t)
\end{align}   
where $l \in \clist{ 1, \ldots, \sorder}$,  $\slambdaset_{\neg l} = \plist{\slambda_1, \ldots, \slambda_{l-1}, \slambda_{l+1}, \ldots, \slambda_\sorder}$, and $\vfunc_{k, \slambdaset}(t) = 0$ for $k < 0$ and $k \geq  \sorder $. 
\end{proposition}
\begin{proof}
See \refapp{app.VandermondeDynamicsRecursion}.
\end{proof}

An important shared property of Vandermonde and exponential basis functions is nonnegativity, as shown in \reffig{fig.ExponentialVandermondeBasisFunctions}.

\begin{proposition}\label{prop.VandermondeBasisNonnegativity}
\emph{(Nonnegative Vandermonde Basis)}
For any real negative $\slambdaset = \plist{\slambda_1, \ldots, \slambda_\sorder} \in \R_{<0}^{\sorder}$ with  distinct elements (i.e., $\slambda_i \neq \slambda_j$ for all $i \neq j$), the Vandermonde basis functions are nonnegative, i.e.,
\begin{align} \label{eq.VandermondeBasisNonnegativity}
\vfunc_{k,\slambdaset}(t) \geq 0 \quad \quad \forall k=0, \ldots, \sorder-1.
\end{align}
\end{proposition}
\begin{proof}
See \refapp{app.VandermondeBasisNonnegativity}.
\end{proof}

Another common characteristic feature of Vandermonde and exponential basis functions are their relative order, see \reffig{fig.ExponentialVandermondeBasisFunctions}.

\begin{proposition}\label{prop.RelativeVandemondeBasisBound}
\emph{(Relative Vandemonde Basis Bounds)}
For any real negative $\slambdaset = \plist{\slambda_1, \ldots, \slambda_\sorder} \in \R_{<0}^{\sorder}$ with  distinct elements (i.e., $\slambda_k \neq \slambda_l$ for all $k \neq l$), the Vandermonde basis functions $\vfunc_{k, \slambdaset}(t)$ in \refeq{eq.VandermondeBasisCramerRule} are relatively bounded by each other in terms of the companion coefficients $\sgain_{k, \slambdaset}$ in \refeq{eq.ControlGain} as \reffn{fn.RelativeVandemondeBasisBound}
\begin{align}\label{eq.RelativeVandemondeBasisBound}
\sgain_{k, \slambdaset_{\neg i}} \vfunc_{k-1, \slambdaset}(t) \geq \sgain_{k-1, \slambdaset_{\neg i}}\vfunc_{k,\slambdaset} (t). 
\end{align}
for any $k = 1, \ldots \sorder -1$ and $i = 1, \ldots, \sorder$, where $\slambdaset_{\neg i} = \plist{\slambda_{1}, \ldots, \slambda_{i-1}, \slambda_{i+1}, \ldots, \slambda_{\sorder}}$.
\end{proposition}
\begin{proof}
See \refapp{app.RelativeVandemondeBasisBound}.
\end{proof}

\addtocounter{footnote}{1}\footnotetext{\label{fn.RelativeVandemondeBasisBound}
Vandermonde basis functions can also be relatively bounded as 
\begin{align}
\sgain_{k, \slambdaset} \vfunc_{k-1, \slambdaset}(t) \geq \sgain_{k-1, \slambdaset}\vfunc_{k,\slambdaset} (t). \nonumber
\end{align}
which is less accurate compared to the tight bound in \refeq{eq.RelativeVandemondeBasisBound}.
For example, for the second-order companion system with negative real eigenvalues $\slambdaset = \plist{\slambda_{\min}, \slambda_{\max}} \in \R_{< 0}^{2}$, one has
\begin{align}
\frac{\vfunc_{0,\slambdaset}}{\vfunc_{1,\slambdaset}} \geq \frac{\sgain_{0, \slambdaset_{\neg\max}}}{\sgain_{1, \slambdaset_{\neg\max}}} = -\slambda_{\min} > \frac{\sgain_{0,\slambdaset}}{\sgain_{1, \slambdaset}} = -\slambda_{\min}\frac{\slambda_{\max}}{\slambda_{\min} + \slambda_{\max}}. \nonumber 
\end{align} 
}

The ordering relation of Vandermonde basis function in  \refprop{prop.RelativeVandemondeBasisBound} is tight at the limit.
\begin{proposition}\label{prop.VandermondeBasisRatioLimit}
\emph{(Vandermonde Basis Ratio Limit)}
For any real negative $\slambdaset = \plist{\slambda_1, \ldots, \slambda_\sorder} \in \R_{<0}^{\sorder}$ with  distinct elements (i.e., $\slambda_i \neq \slambda_j$ for all $i \neq j$), the relative Vandermonde bound $\sgain_{k, \slambdaset_{\neg \max}}\vfunc_{k-1,\slambdaset}(t) \geq \sgain_{k-1, \slambdaset_{\neg \max}}\vfunc_{k, \slambdaset}(t)$ is tight and becomes an equality as $t \rightarrow \infty$ , i.e,
\begin{align}
\lim_{t \rightarrow \infty} \frac{\vfunc_{k-1, \slambdaset}(t)}{\vfunc_{k, \slambdaset}(t)} = \frac{\sgain_{k-1, \slambdaset_{\neg \max}}}{\sgain_{k, \slambdaset_{\neg \max}}}
\end{align}
\end{proposition}
\begin{proof}
See \refapp{app.VandermondeBasisRatioLimit}
\end{proof}

Last but not least, both exponential and Vandermonde basis functions are bounded above,  as illustrated in \reffig{fig.ExponentialVandermondeBasisFunctions}.

\begin{proposition}\label{prop.ZerothVandermondeBasisFunction}
\emph{(Vandermonde Basis Upper Bound)}
For any real negative $\slambdaset = \plist{\slambda_1, \ldots, \slambda_\sorder} \in \R_{<0}^{\sorder}$ with  distinct elements (i.e., $\slambda_k \neq \slambda_l$ for all $k \neq l$), the $0^{th}$ Vandermonde basis function is tightly upper bounded by 1, i.e.,
\begin{align}
\vfunc_{0,\slambdaset}(t) \leq 1 \quad \text{and}\quad \dot{\vfunc}_{0,\slambdaset}(t) \leq 0  \quad \forall  t \geq 0.
\end{align}
where equality holds at $t=0$, i.e., $\vfunc_{0,\slambdaset}(0)=1$.  
\end{proposition}
\begin{proof}
See \refapp{app.ZerothVandermondeBasisFunction}.
\end{proof}

\section{Explicit Convex Trajectory Bounds \\ for Linear Companion Systems}
\label{sec.ConvexCompanionTrajectoryBounds}

In this section, we describe how to construct convex simplicial trajectory bounds for linear companion systems using the common characteristic features of Vandermonde and exponential basis functions, as an accurate alternative to ellipsoidal invariant sublevel sets built based on the Lyapunov theory.

\subsection{Simplicial Trajectory Bounds for Companion Systems}

Three important shared properties of Vandermonde and exponential basis functions are nonnegativity, relative ordering, and boundedness.
This allows for constructing convex simplicial trajectory bounds on companion system motion based on a convex combination of trajectory coefficients, which is inspired by the convexity of B\'ezier curves \cite{arslan_tiemessen_TRO2022}. 

\begin{proposition}\label{prop.SimplicialTrajectoryBound}
\emph{(Simplicial Companion Trajectory Bounds)}
Suppose the solution trajectory $\sposn(t)$ of the $\sorder^{\text{th}}$-order companion dynamics in \refeq{eq.CompanionDynamics}, starting at $t = 0$ from an initial state $\sstate_0 = (\sposn^{(0)}_0, \ldots, \sposn^{(\sorder-1)}_0) \in \R^{\sorder \times \sdim}$, can be expressed  using some scalar basis functions $\theta_{0}(t), \ldots, \theta_{\sorder-1}(t)$ as  
\begin{align}
\sposn(t) = \sum_{i=0}^{\sorder-1} \vect{y}_i(\sstate_0) \theta_{i} (t)
\end{align}
where $\vect{y}_0 (\sstate_0), \ldots, \vect{y}_{\sorder -1} (\sstate_0) \in \R^{\sdim}$ are trajectory coefficients depending on the initial state $\sstate_0$, and the basis functions $\theta_{0}(t), \ldots, \theta_{\sorder-1}(t)$ satisfy
\begin{enumerate}[\quad i)]
\item (Nonnegativity) $\theta_{i}(t) \geq 0 $ for all $i$,
\item (Relative Ordering) $\beta_i \theta_{i}(t) \geq \beta_{j} \theta_{j}(t)$ for all $i \leq j$ with some fixed positive scalars $\beta_i, \beta_j > 0$, 
\item (Boundedness) $\theta_0 (t) \leq 1$.
\end{enumerate}
Then, the companion system trajectory $\sposn(t)$ is bounded for all $t \geq 0$ by a simplicial\footnote{Our naming convention is inspired by the fact that if the companion system is moving in the $\sorder$-dimensional Euclidean space, i.e., $\sdim = \sorder$, then the motion prediction in \refeq{eq.SimplicialTrajectoryBound} corresponds to a simplex (i.e., a hyper-triangle).} convex region determined by trajectory coefficients $\vect{y}_0 (\sstate_0), \ldots, \vect{y}_{\sorder -1} (\sstate_0)$ as~\reffn{fn.SimplicialTrajectoryBound}
\begin{align}\label{eq.SimplicialTrajectoryBound}
\sposn(t) \in \conv \plist{\sum_{j=0}^{i-1} \frac{\beta_0}{\beta_j} \vect{y}_j(\sstate_0)  \Bigg | i = 0, \ldots, \sorder}
\end{align} 
where $\conv$ denotes the convex hull operator.
\end{proposition}
\begin{proof}
See \refapp{app.SimplicialTrajectoryBound}.
\end{proof}

\addtocounter{footnote}{1}\footnotetext{\label{fn.SimplicialTrajectoryBound}
Summation over the empty set is assumed to be zero. Hence, the simplicial companion trajectory bound in \refeq{eq.SimplicialTrajectoryBound} has the following form
\begin{align*}
\conv \plist{0,  \sum_{j=0}^{0} \frac{\beta_0}{\beta_j} \vect{y}_j(\sstate_0), \ldots, \sum_{j=0}^{\sorder-1} \frac{\beta_0}{\beta_j} \vect{y}_j(\sstate_0)}.
\end{align*}
}

In addition to applying \refprop{prop.SimplicialTrajectoryBound} for Vandermonde basis functions in \refthm{thm.VandermondeSimplex}, we construct exponential simplexes for bounding the companion trajectory using the exponential basis function, as illustrated in \reffig{fig.VandermondeExponentialSimplexEigenValue}.

\begin{proposition} \label{prop.ExponentialSimplex}
\emph{\!(\mbox{Exponential Simplexes for Companion Systems})}
For any real negative distinct ordered $\slambdaset=(\slambda_1, \ldots, \slambda_\sorder) \in \R^{\sorder}$ with $\slambda_1 < \slambda_2 < \ldots < \slambda_\sorder < 0$, the solution trajectory  $\sposn(t)=\sum_{j=0}^{\sorder-1} \scoef_{j, \slambdaset}(\sstate_0) e^{\slambda_j t}$ of the companion dynamics in \refeq{eq.CompanionDynamics} starting at $t=0$ from $\sstate_0$ is contained in
\begin{align}\label{eq.ExponentialSimplex}
\sposn(t) \in \esplx_{\slambdaset}(\sstate_0):= \conv \plist{\sum_{j=1}^{i} \scoef_{j,\slambdaset}(\sstate_0) \Bigg | i = 0, \ldots, \sorder}
\end{align}  
where $\scoef_{1,\slambdaset}(\sstate_0), \ldots, \scoef_{\sorder, \slambdaset}(\sstate_0)$ are defined as in \refeq{eq.ExponentialTrajectoryCoefficients}. 
\end{proposition}
\begin{proof}
The result directly follows from \refprop{prop.SimplicialTrajectoryBound} and the fact that $0 \leq e^{\slambda_1 t}< \ldots <e^{\slambda_\sorder t} \leq 1$.
\end{proof} 

As seen in \reffig{fig.VandermondeExponentialSimplexEigenValue}, Vandermonde and exponential simplexes are strongly related to each other because their vertices are related by a linear transformation which is singular for identical characteristic polynomial roots due to \refeq{eq.ExponentialTrajectoryCoefficients}.

\begin{figure}[t]
\begin{tabular}{@{}c@{\hspace{0.5mm}}c@{\hspace{0.5mm}}c@{}}
\includegraphics[width=0.16\textwidth]{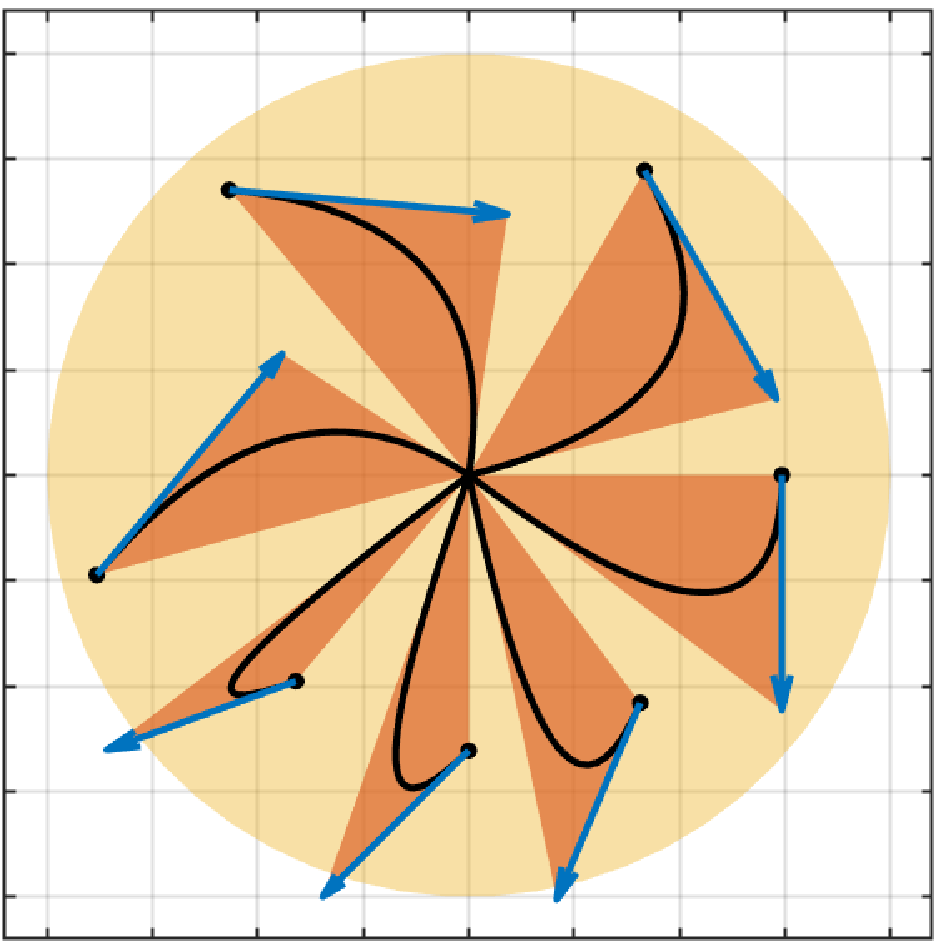} 
&
\includegraphics[width=0.16\textwidth]{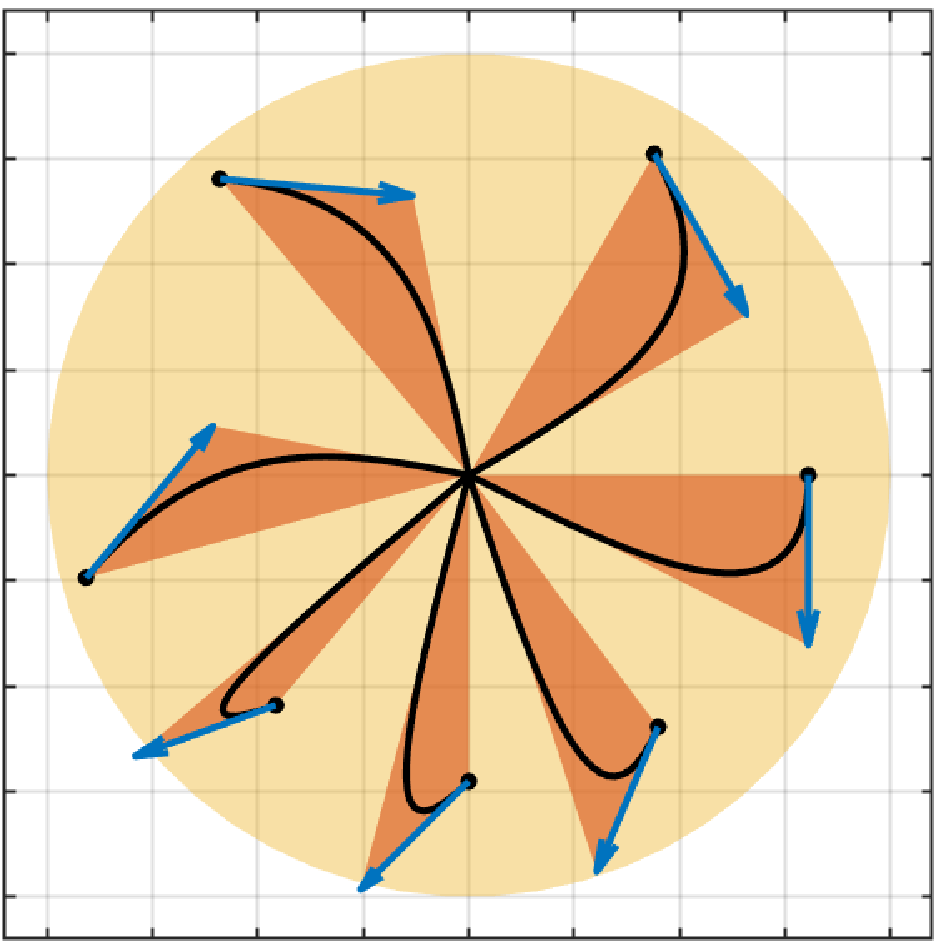} 
&
\includegraphics[width=0.16\textwidth]{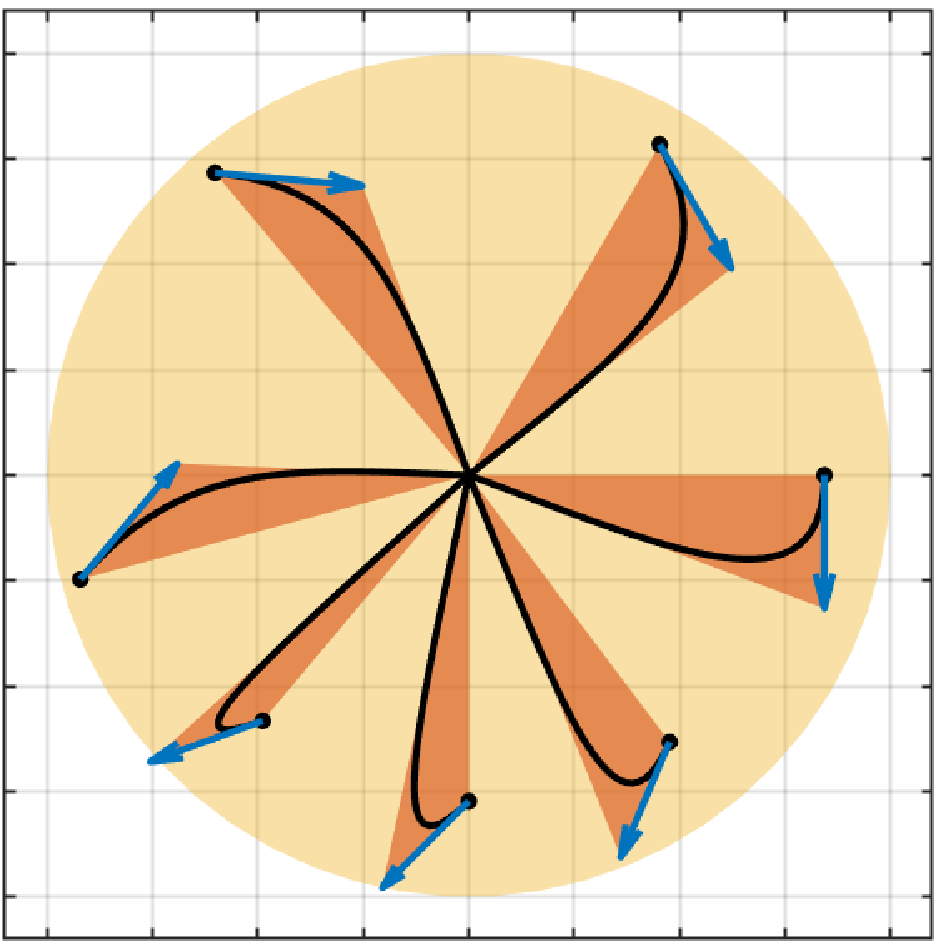} 
\\
\includegraphics[width=0.16\textwidth]{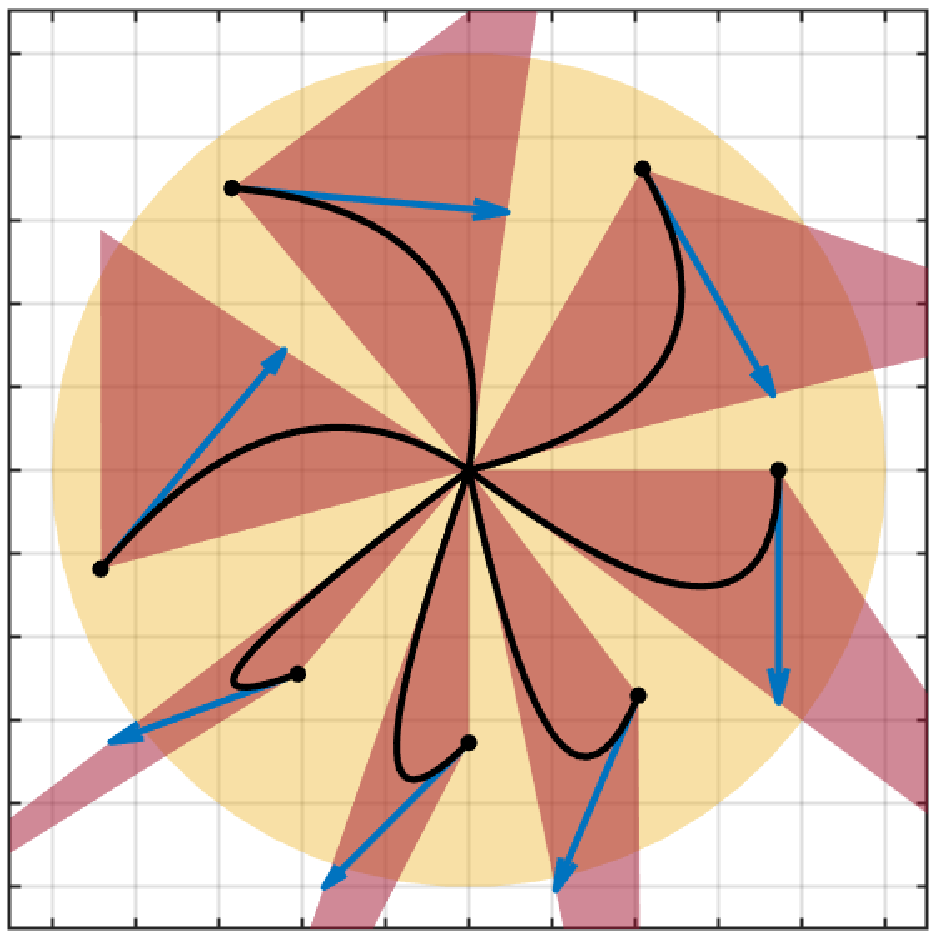} 
&
\includegraphics[width=0.16\textwidth]{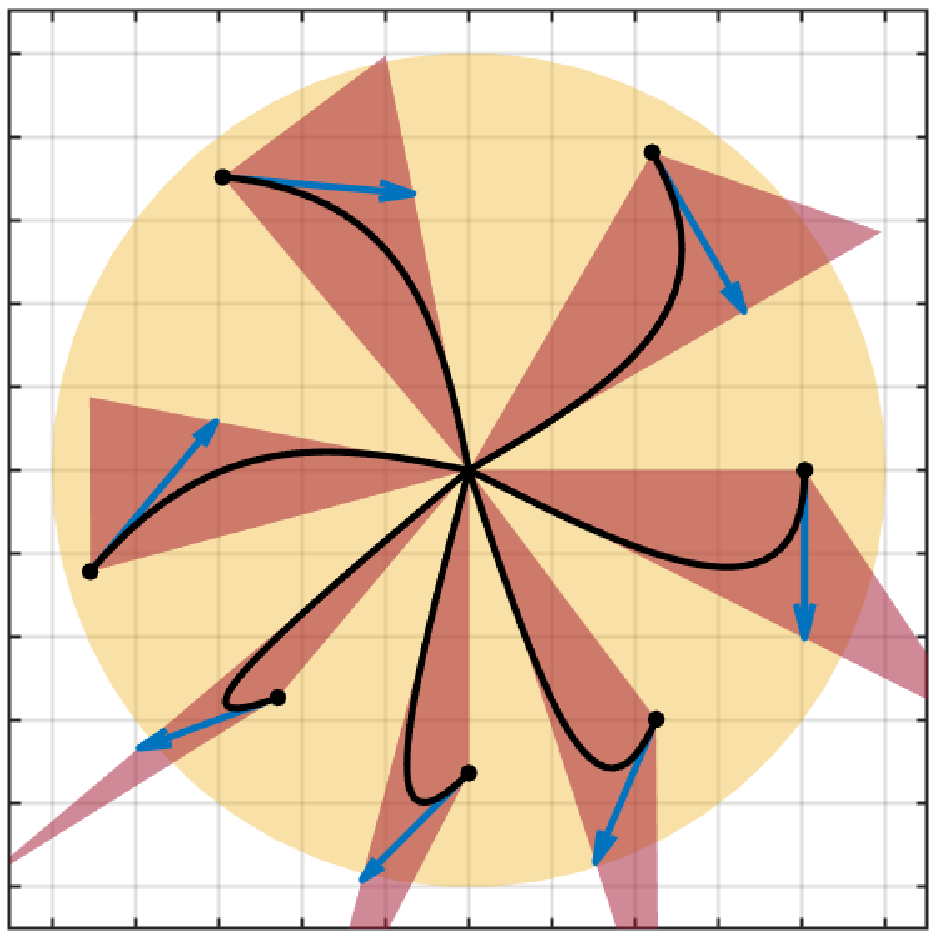} 
&
\includegraphics[width=0.16\textwidth]{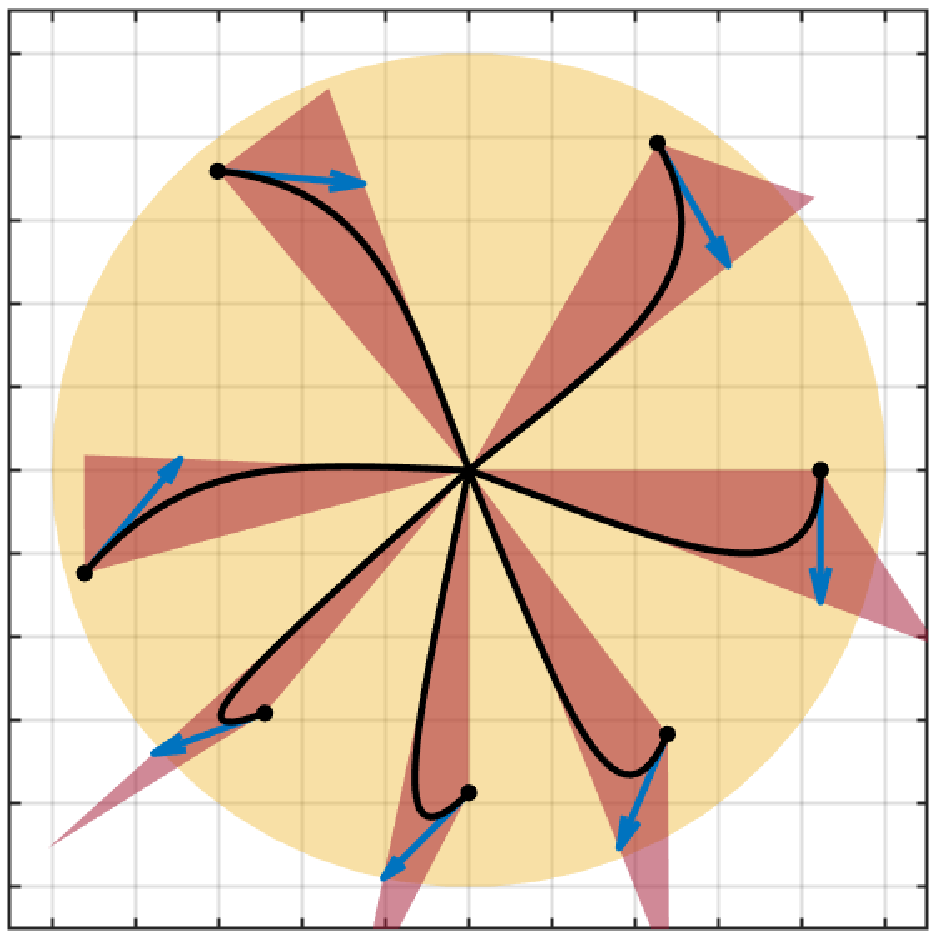} 
\end{tabular}
\caption{The effect of characteristic polynomial roots on (top) Vandermonde and (bottom) exponential simplexes that bound the motion trajectory of the second-order companion system: (left) $\slambdaset = (\lambda_1, \lambda_2)  = (-1, -2)$, (middle) $\slambdaset=(\lambda_1, \lambda_2) = (-1, -3)$, and (right) $\slambdaset = (\lambda_1, \lambda_2) = (-1,-4)$. Lyapunov ellipsoids (yellow) are constructed using $\mat{D} = \mat{I}_{\sorder \sdim \times \sorder \sdim}$, and all initial system states have the same Lyapunov level value. The system starts from an initial position denoted by a black dot and with a velocity shown by a blue arrow that is scaled by a factor of $\frac{\sgain_{1, \slambdaset_{\neg \max}}}{\sgain_{0, \slambdaset_{\neg \max}}} = \frac{1}{| \min(\slambdaset)| }$. }
\label{fig.VandermondeExponentialSimplexEigenValue}
\end{figure}

\vspace{-3mm}

\subsection{Projected Lyapunov Ellipsoids for Companion Systems}
\label{sec.ProjectedLyapunovEllipsoid}

To demonstrate the significance and accuracy of the proposed simplicial companion trajectory bounds, we consider invariant Lyapunov level sets that are widely used for constrained control and optimization of dynamical systems \cite{blanchini_Automatica1999}.
Using the Lyapunov stability \cite{chen_LinearSystemTheory1998} of linear state-space companion dynamics in \refeq{eq.StateSpaceCompanionDynamics}, a simple analytic ellipsoidal trajectory bound for companion systems can be built based on orthogonal projections of invariant sublevel sets of a quadratic Lyapunov function, as illustrated in \reffig{fig.LyapunovEllipsoidsDecayMatrix}.

\begin{proposition} \label{prop.LyapunovEllipsoid}
\emph{(Projected Lyapunov Ellipsoids)}
For any real Hurwitz companion matrix $\cmat_{\slambdaset}$ associated with eigenvalues $\slambdaset = (\slambda_1, \ldots, \slambda_\sorder) \in \C^{\sorder}$ with strictly negative real parts, the solution trajectory $\sposn(t)$ of the companion dynamics in \refeq{eq.CompanionDynamics} starting at $t =0$ from any initial state $\sstate_0 \in \R^{\sorder \sdim}$ is contained  in the projected Lypunov ellipsoid that is defined as
\begin{align}
\sposn(t) \in \elp\plist{\mat{0}, \mat{I}_{\sdim \times \sorder \sdim} \mat{P}^{-1} \mat{I}_{\sorder \sdim \times \sdim}, \norm{\sstate_0}_{\mat{P}}} \quad \forall t \geq 0
\end{align}
where $\elp(\elpctr, \elpmat, \elprad) \!:=\!\! \clist{\!\elpctr \!+ \!\elprad \elpmat^{\frac{1}{2}} \vect{x} \Big| \vect{x} \!\in\! \R^{\sdim}\!, \norm{\vect{x}} \!\leq\! 1\!}$ is the ellipsoid, centered at $\elpctr \in \R^{\sdim}$ and associated with a positive semidefinite matrix\reffn{fn.MatrixSquareRoot} $\elpmat \in \PSDM^{\sdim}$ and a nonnegative scalar $\elprad \geq 0$, and $\mat{P} \in \PDM^{\sorder\sdim}$ is a symmetric positive-definite Lyapunov matrix that uniquely satisfies the Lyapunov equation   
\begin{align}\label{eq.CompanionLyapunovEquation}
\tr{(\cmat_{\slambdaset}\otimes \mat{I}_{\sdim \times \sdim})} \mat{P} + \mat{P} (\cmat_{\slambdaset}\otimes \mat{I}_{\sdim \times \sdim}) + \tr{\mat{D}}\mat{D} = \mat{0} 
\end{align}
for some decay matrix $\mat{D} \in \R^{m \times \sorder \sdim}$  such that $(\cmat_{\slambdaset}\otimes \mat{I}_{\sdim \times \sdim})$ is observable.
Here, $\otimes$ denotes the Kronecker product,  $\mat{I}_{\sorder \sdim \times \sdim}$ is the $\sorder \sdim \times \sdim$ dimensional identity matrix, and $\norm{\vect{x}}_{\mat{W}} := \sqrt{\tr{\vect{x}} \mat{W} \vect{x}}= \norm{\mat{W}^{\frac{1}{2}} \vect{x}}$ is the weighted Euclidean norm associated with a positive definite matrix $\mat{W} \in \PDM^{n}$ and  $\norm{.}$ denotes the standard Euclidean norm.
\end{proposition}
\begin{proof}
See \refapp{app.LyapunovEllipsoid}.
\end{proof}

As expected, the accuracy of Lyapunov ellipsoids depends on the selection of the decay matrix $\mat{D}$.
In our numerical studies, we observe that the identity decay matrix $\mat{D} = \mat{I}_{\sorder \sdim \times \sorder \sdim}$ gives the best performance, as  seen in Figures \ref{fig.VandermondeExponentialSimplexEigenValue}-\ref{fig.LyapunovEllipsoidsDecayMatrix}.

\begin{figure}[t]
\begin{tabular}{@{}c@{\hspace{0.5mm}}c@{\hspace{0.5mm}}c@{}}
\includegraphics[width=0.16\textwidth]{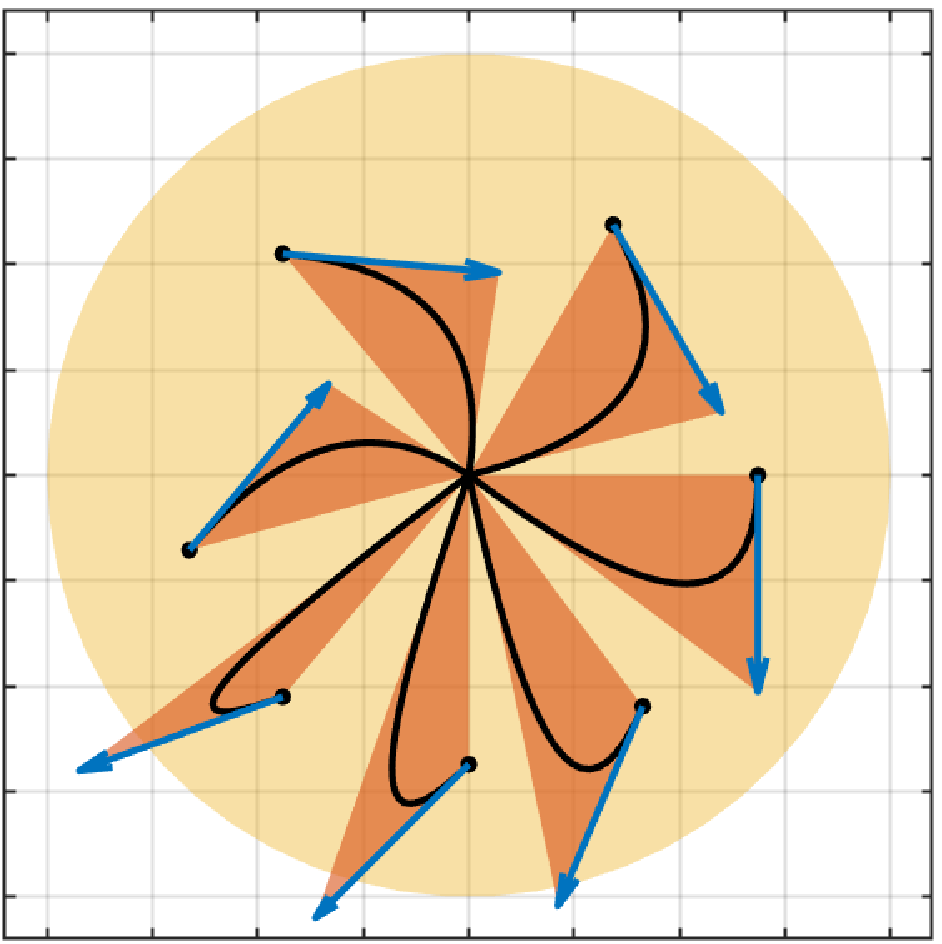} 
& 
\includegraphics[width=0.16\textwidth]{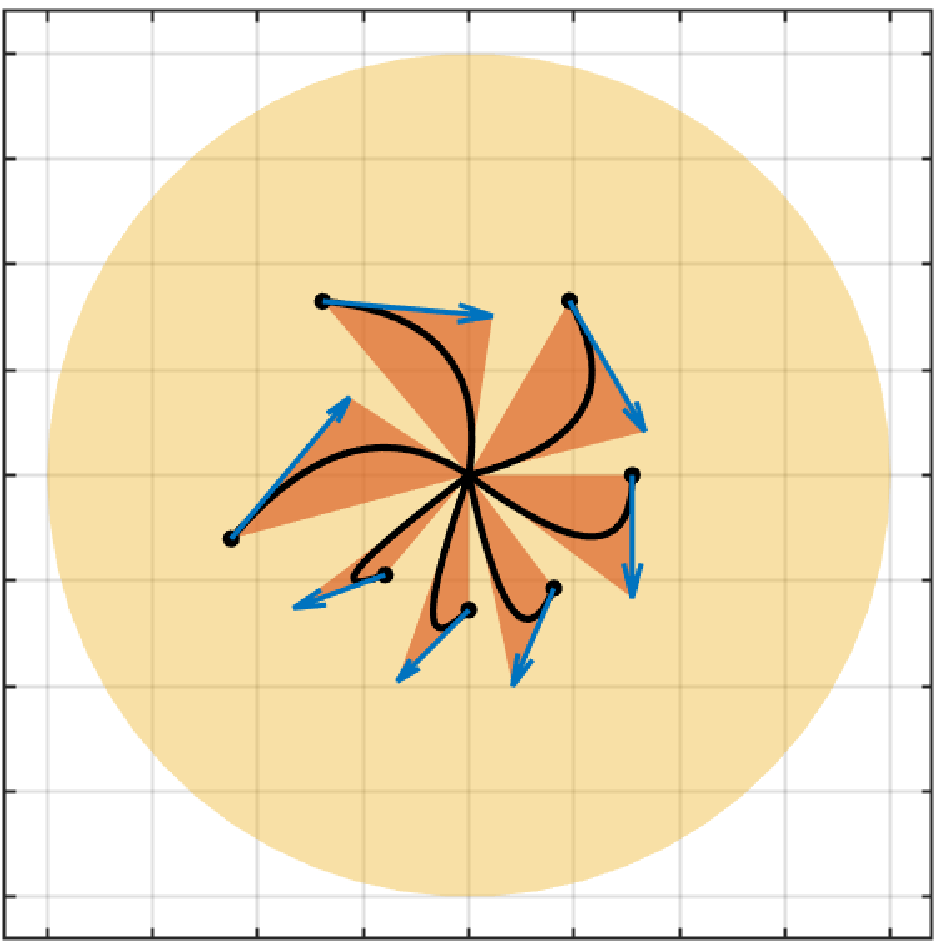} 
&
\includegraphics[width=0.16\textwidth]{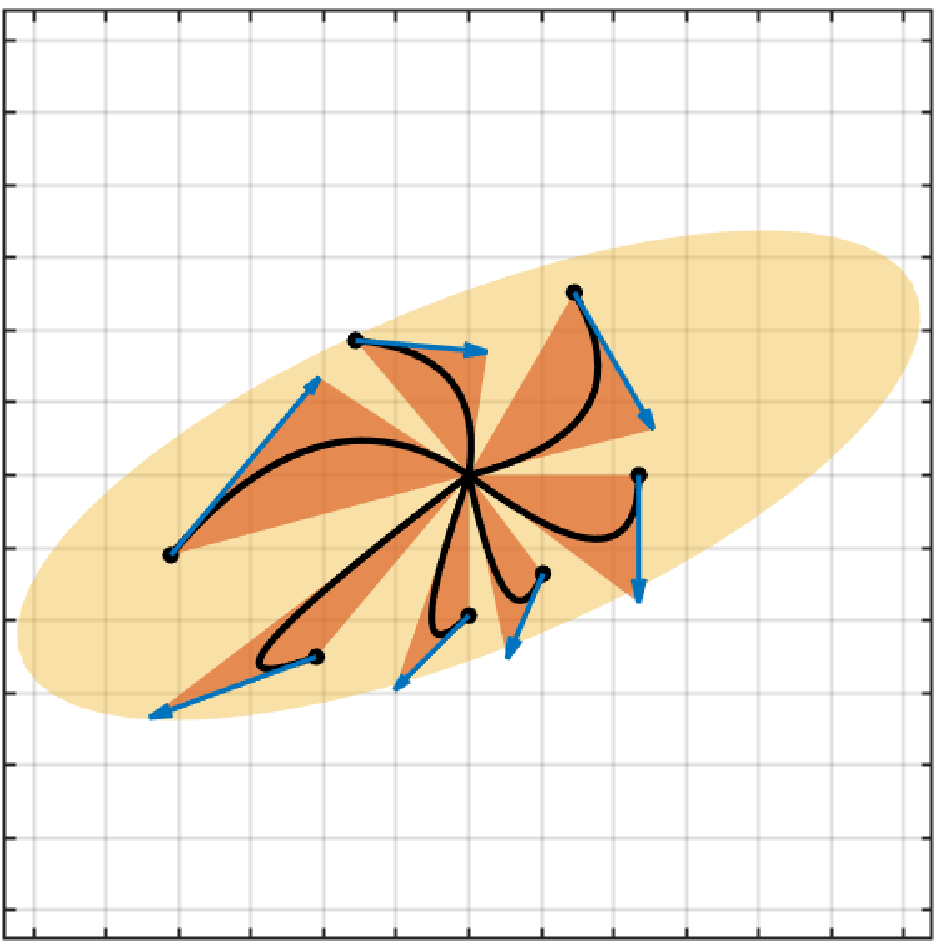} 
\end{tabular}
\caption{The effect of the decay matrix $\mat{D}$ on the projected Lyapunov ellipsoids (yellow) that bound the motion trajectory of the second-order companion system with $\slambdaset =(\lambda_1, \slambda_2) = (-1, -2)$: (left) $\mat{D} = \diag([0,1]) \otimes \diag([1, 1])$, (middle) $\mat{D} = \diag([1,0]) \otimes \diag([1, 1])$, (right) $\mat{D} = \diag([0,1]) \otimes \scalebox{0.6}{$\begin{bmatrix}1.0 & -0.5 \\ -0.5 & 2.0\end{bmatrix}$})$, and see \reffig{fig.VandermondeExponentialSimplexEigenValue} (top, left) for $\mat{D} = \diag([1,1]) \otimes \diag([1, 1])$.
The second-order companion system starts from an initial position denoted by a black dot with a velocity shown by a blue arrow at a shared Lyapunov level value. 
Vandermonde simplexes (orange) capture the companion motion more accurately compared to Lyapunov ellipsoids (yellow).}
\label{fig.LyapunovEllipsoidsDecayMatrix}
\end{figure}

\addtocounter{footnote}{1}\footnotetext{\label{fn.MatrixSquareRoot}The unique symmetric positive-definite square-root of $\elpmat$ that satisfies $\elpmat = \elpmat^{\frac{1}{2}} \tr{(\elpmat^{\frac{1}{2}})\!}$ can be determined as $\elpmat^{\frac{1}{2}}\!=\!\mat{V}\diag\plist{\sqrt{\sigma_1}, \ldots, \sqrt{\sigma_\sorder}} \tr{\mat{V}}$ using the singular value decomposition $\elpmat = \mat{V} \diag\plist{\sigma_1, \ldots, \sigma_n} \tr{\mat{V}} $, where $\diag\plist{\sigma_1, \ldots, \sigma_n}$ is the diagonal matrix with  elements ${\sigma_1, \ldots, \sigma_n}$.}

\section{Numerical Simulations}
\label{sec.NumericalSimulations}

In this section, we systematically investigate the accuracy of Vandermonde and exponential simplexes compared to Lyapunov ellipsoids in capturing companion system motion in extensive numerical simulations. 
We consider high-order companion systems with different differential order $\sorder$ moving in the 2D and 3D Euclidean spaces starting from random initial states that are uniformly sampled over the unit hyper-cube $[-1, 1]^{\sorder \sdim}$.
As a performance measure, we use the area/volume ratio of Vandermonde and exponential simplexes relative to Lyapunov ellipsoids.
We set the decay matrix as the identity matrix $\mat{D} = \mat{I}_{\sorder \sdim \times \sorder \sdim}$ because it yields the best performance for Lyapunov ellipsoids as seen in Figures \ref{fig.VandermondeExponentialSimplexEigenValue}-\ref{fig.LyapunovEllipsoidsDecayMatrix}.

\begin{figure}[b]
\centering
\begin{tabular}{@{}c@{}c@{}}
\includegraphics[width=0.49\columnwidth]{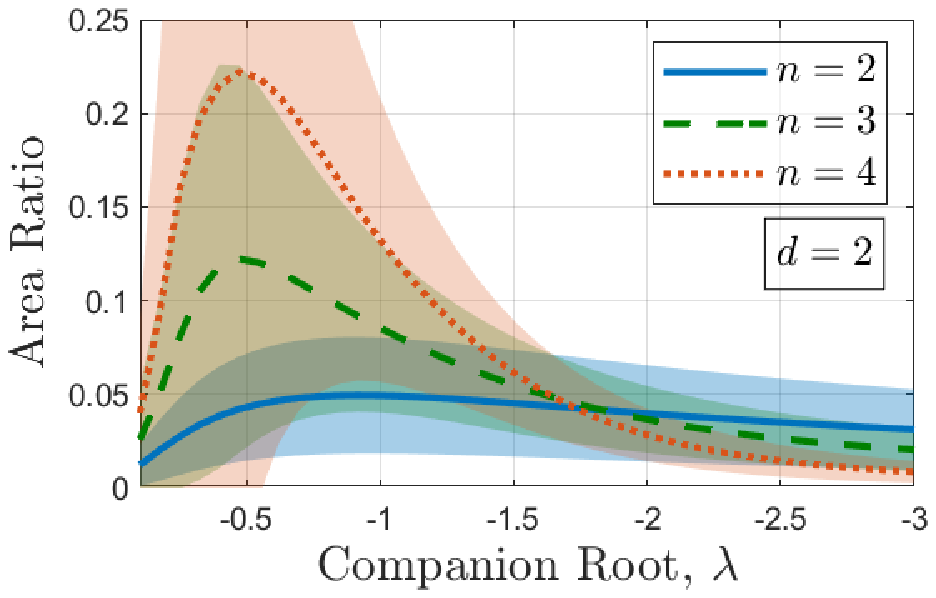}
&
\includegraphics[width=0.49\columnwidth]{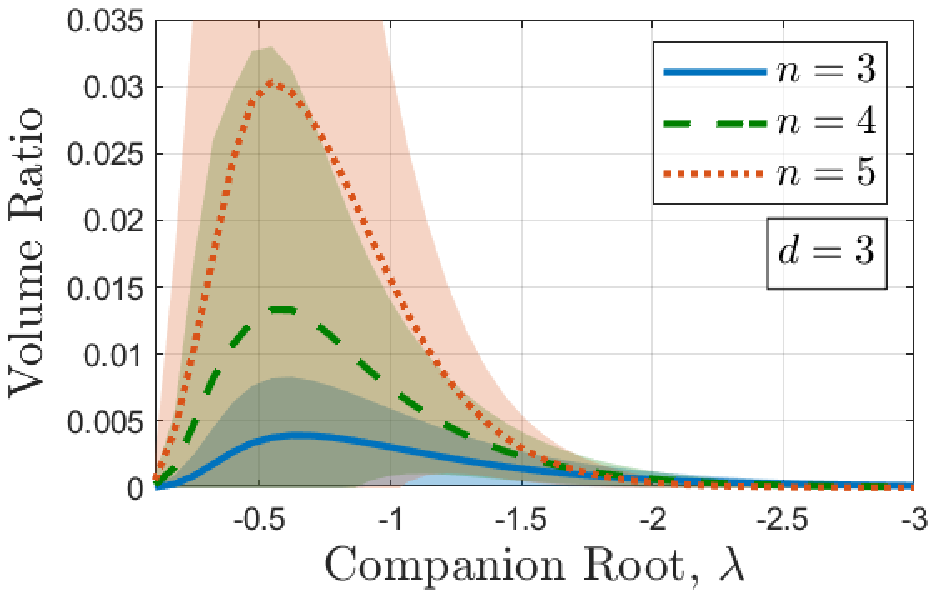}
\\
\includegraphics[width=0.49\columnwidth]{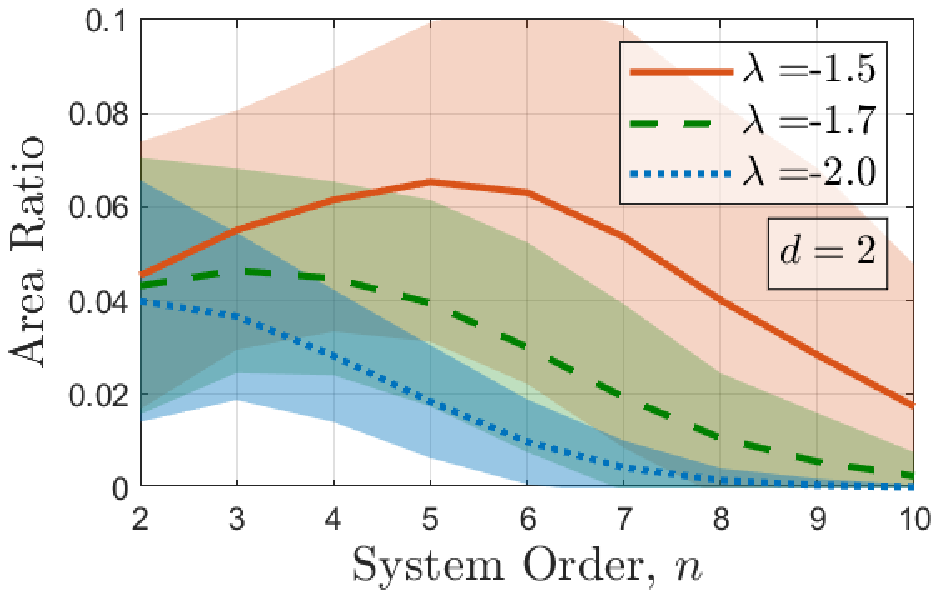}
& 
\includegraphics[width=0.49\columnwidth]{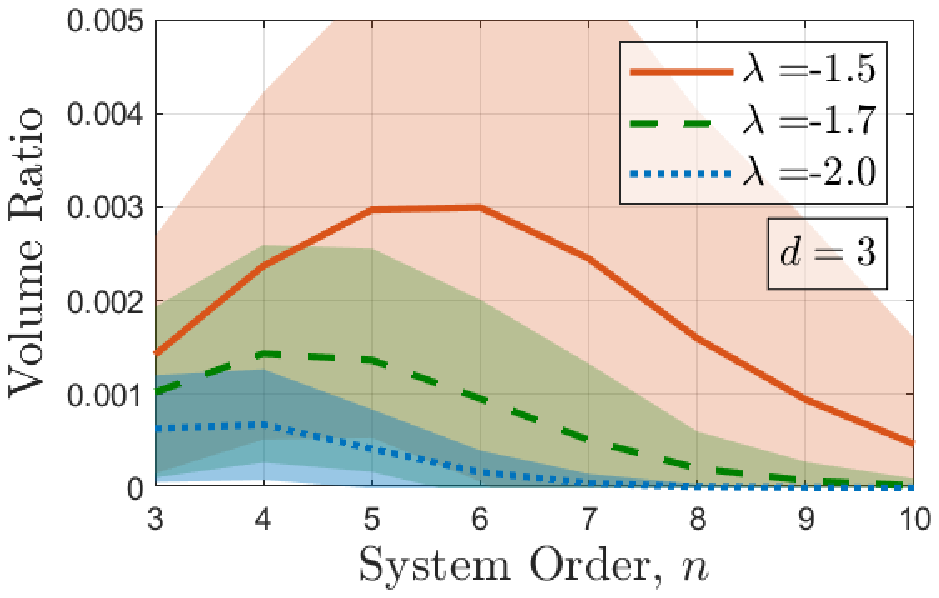}
\end{tabular}
\vspace{-2mm}
\caption{Accuracy (area/volume ratio) of Vandemonde simplexes relative to Lyapunov ellipsoids versus (top) the companion characteristic polynomial roots and (bottom) companion system order in the (left) 2D and (right) 3D Euclidean space. Here, the companion characteristic polynomial roots are assumed to be identical and the decay matrix is set to be $\mat{D} = \mat{I}_{\sorder \sdim \times \sorder \sdim}$.}
\label{fig.VandermondeSimplexLyapunovEllipsoidAccuracy}
\end{figure}

We first compare the accuracy of Vandermonde simplexes relative to Lyapunov ellipsoids in \reffig{fig.VandermondeSimplexLyapunovEllipsoidAccuracy} by investigating the role of characteristic polynomial roots and system order on their performance. 
We observe that Vandermonde simplexes outperform Lyapunov ellipsoids in defining an accurate bound on companion trajectory with a smaller volume/area. 
For example, Vandermonde simplexes offer at least an order (respectively, two orders) of magnitude improvement for the area (volume, respectively) of Lyapunov ellipsoids on average for companion systems when the characteristic polynomial roots are larger than unity in magnitude. 
We also see that the accuracy of Vandermonde simplexes relative to Lyapunov ellipsoids is often increasing with increasing characteristic polynomial root and system order when the magnitude of the characteristic polynomial root is larger than unity.

To compare all three convex companion trajectory bounds with respect to each other in \reffig{fig.CompanionTrajectoryBoundAccuracy}, we consider companion systems whose characteristic roots are distinct and uniformly distributed over the range $[-0.5, \slambda]$ for a selection of $\slambda$ in  $[-3.0, -0.1]$. We particularly select the range $[-0.5, \slambda]$ to preserve the peak performance of Lyapunov ellipsoids relative to Vandermonde simplexes observed in \reffig{fig.VandermondeSimplexLyapunovEllipsoidAccuracy}. 
We again observe that the accuracy (area/volume ratio) of Vandermonde simplexes relative to Lyapunov ellipsoids is increasing with the increasing characteristic polynomial root $\slambda$ for $|\slambda| \geq 0.5$.
Exponential simplexes show a similar trend like Vandermonde simplexes away from their singularity where the difference between the characteristic polynomial roots becomes small.
As expected, exponential simplexes are uninformative around their singularity. 
The performance of exponential simplexes significantly depends on the system order and the space dimension.
An interesting observation in \mbox{\reffig{fig.CompanionTrajectoryBoundAccuracy}~(bottom)} is that Vandermonde and exponential simplexes are strongly related to each other when the system order and the space dimension are the same (i.e., $\sorder = \sdim$), which can be explained by that the fact that  both Vandermonde and exponential simplexes are simplex (i.e., hyper-triangle) shaped for $\sorder = \sdim$, see \reffig{fig.VandermondeExponentialSimplexEigenValue}, and their vertices are related by a linear transformation, see \refeq{eq.VandermondeSimplex} and \refeq{eq.ExponentialSimplex}. 
In summary, we conclude that Vandermonde simplexes significantly outperform exponential simplexes and Lyapunov ellipsoids in capturing companion system motion because they do not suffer from the conservatism of invariant Lyapunov sets and the singularity of exponential simplexes.

\begin{figure}[t]
\centering
\begin{tabular}{@{}c@{}c@{}}
\includegraphics[width=0.49\columnwidth]{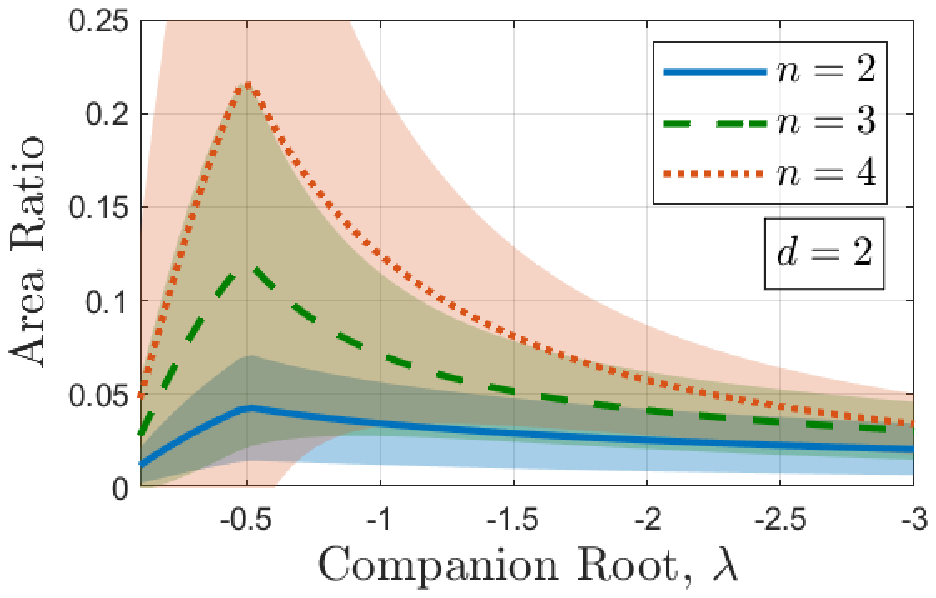} 
&
\includegraphics[width=0.49\columnwidth]{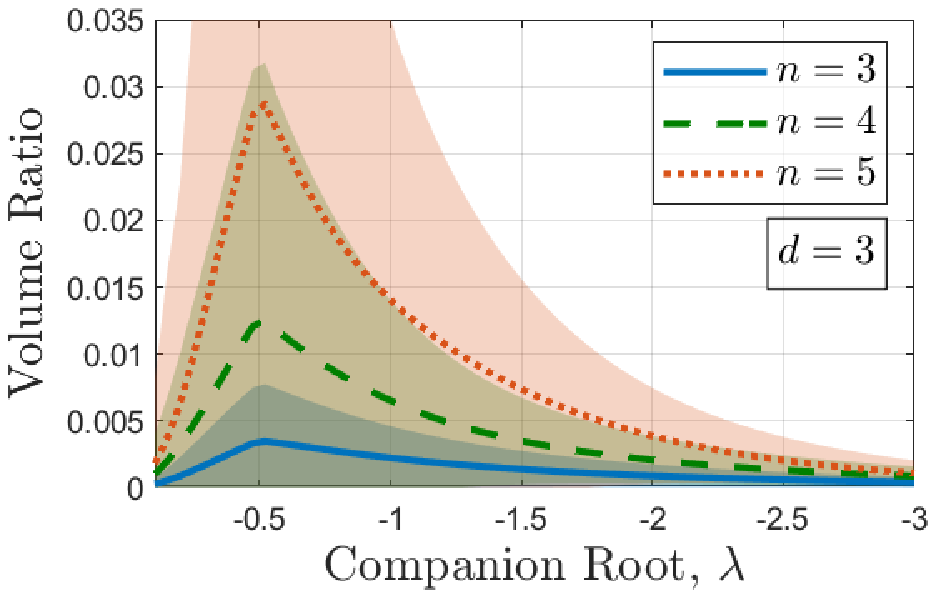}
\\ 
\includegraphics[width=0.49\columnwidth]{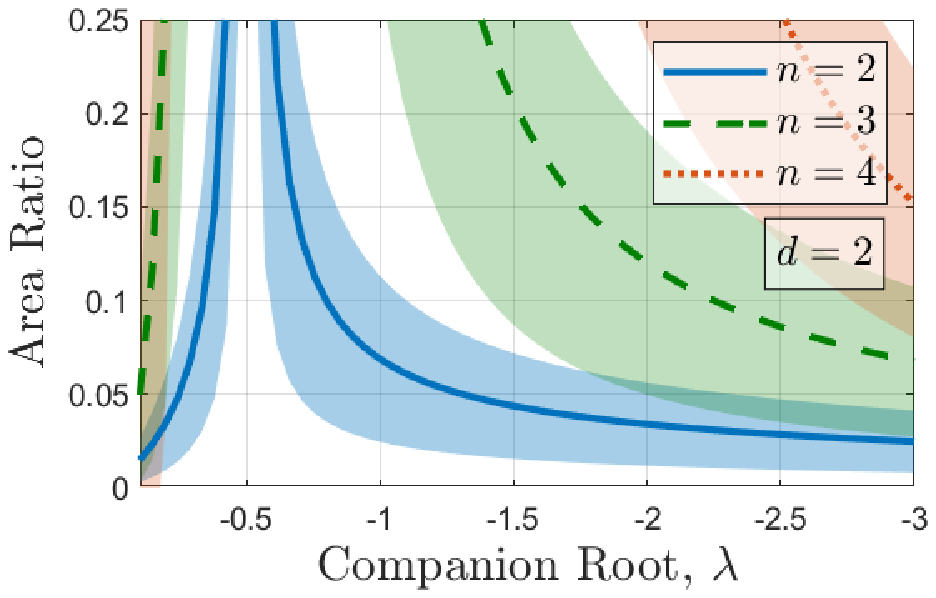} 
&
\includegraphics[width=0.49\columnwidth]{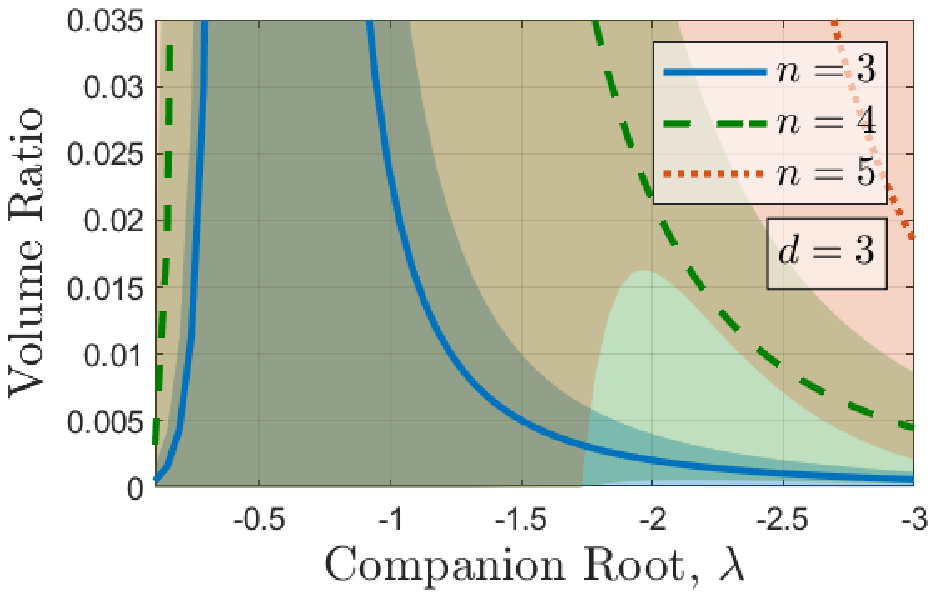} 
\\
\includegraphics[width=0.49\columnwidth]{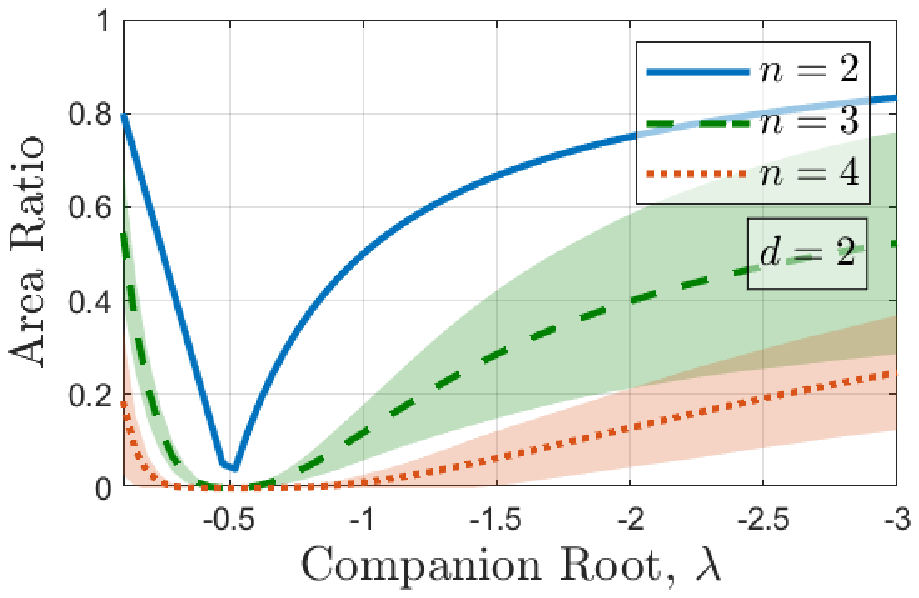} 
&
\includegraphics[width=0.49\columnwidth]{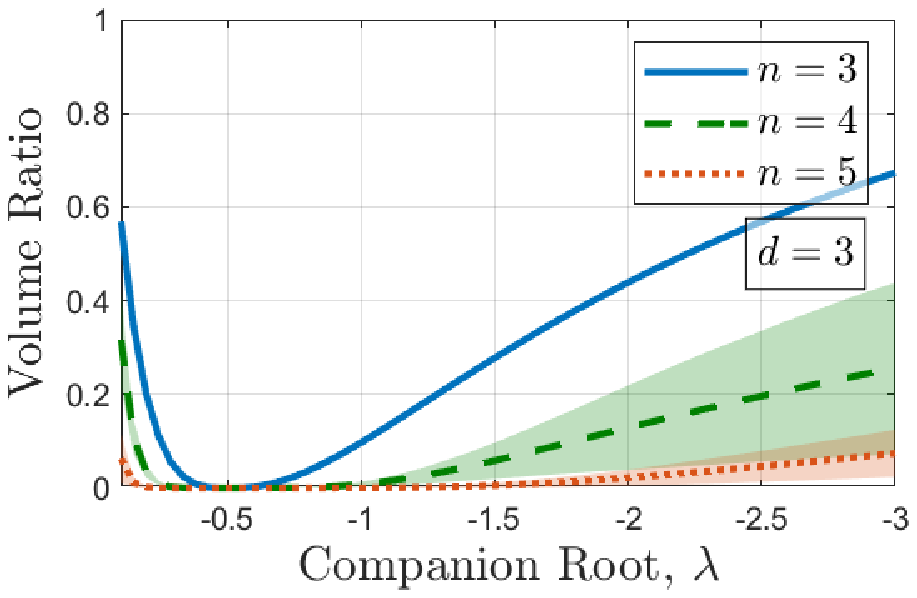}  
\end{tabular}
\vspace{-2mm}
\caption{Accuracy (area/volume ratio) of convex companion trajectory bounds in the (left) 2D and (right) 3D Euclidean space: (top) Vandermonde Simplex vs. Lyapunov Ellipsoid, (middle) Exponential Simplex vs. Lyapunov Ellipsoid, (bottom) Vandermonde Simplex vs. Exponential Simplex.
Here, the companion characteristic polynomial roots are assumed to be uniformly distributed over $[-0.5, \slambda]$, and the decay matrix is set to be $\mat{D} = \mat{I}_{\sorder \sdim \times \sorder \sdim}$.}
\label{fig.CompanionTrajectoryBoundAccuracy}
\end{figure}

\section{Conclusion}
\label{sec.Conclusion}

In this paper, we present two simple analytic complex trajectory bounds, Vandermonde and exponential simplexes, for companion systems that can be used for fast and accurate safety assessment in motion planning and control of nonlinear systems via feedback linearization.
As an alternative to the standard exponential basis functions, we introduce a novel family of Vandermonde basis functions for representing  and understanding companion motion trajectory that allows for new insights into companion system motion.
Using the common characteristic features (nonnegativity, relative ordering, and boundedness) of Vandermonde and exponential basis functions, we introduce a generic procedure to construct convex simplicial trajectory bounds for companion systems, such as Vandermonde and exponential simplexes. 
In extensive numerical simulations, we demonstrate that Vandermonde simplexes show superior performance in describing companion motion with lower spatial complexity (i.e., area and volume) compared to exponential simplexes and Lyapunov ellipsoids.       
We believe that this is yet another evidence supporting that invariant (Lyapunov) sets are conservative for motion prediction, safety assessment, and constraint satisfaction in motion planning and control.
Hence, the design of accurate motion bounds by abandoning set invariance but exploiting motion representation and motion control has significant future promises.

Our current work in progress focuses on accurate and fast feedback motion prediction (i.e., finding motion trajectory bounds) for complex robotic systems (e.g., using Vandermonde simplexes and feedback linearization) and their application for safe robot  motion planning and control \cite{isleyen_vandewouw_arslan_RAL2022, isleyen_vandewouw_arslan_arXiv2022} and path-following control \cite{arslan_arXiv2022}.
A promising research direction is applying Vandermonde simplexes in reachability analysis and model predictive control as a terminal condition for ensuring recursive feasibility with reduced conservatism compared to Lyapunov-like invariant terminal conditions.


\appendices 

\section{Proofs}
\label{sec.Proofs}

\subsection{Proof of \refthm{thm.VandermondeSimplex}}
\label{app.VandermondeSimplex}

\begin{proof}
The result is a direct consequence of \refprop{prop.SimplicialTrajectoryBound} because the companion trajectory can be explained in the Vandermone basis as $\sposn(t) = \sum_{k=0}^{\sorder-1} \vfunc_{k, \slambdaset}(t) \sposn_0^{(k)}$ (see \refprop{prop.CompanionTrajectoryVandermondeBasis}) and the nonnegative  Vandermonde basis functions (see \refprop{prop.VandermondeBasisNonnegativity}) are ordered as $\sgain_{k,\slambdaset_{\neg \max}} \vfunc_{k-1, \slambdaset} (t) \geq \sgain_{k-1, \slambdaset_{\neg \max}} \vfunc_{k,\slambdaset}(t)$ (see \refprop{prop.RelativeVandemondeBasisBound}) and bounded as $\vfunc_{0,\slambdaset}(t) \leq 1$ (\refprop{prop.ZerothVandermondeBasisFunction}).
Note that the characteristic polynomial roots do not need to be distinct due to the continuity of companion system dynamics and Vandermonde simplex vertices with respect to the characteristic roots.
\end{proof}

\subsection{Proof of \refprop{prop.VandermondeDynamicsRecursion}}
\label{app.VandermondeDynamicsRecursion}

\begin{proof}
Using the explicit form of the Vandermonde basis function $\vfunc_{k,\slambdaset}(t)$ in \refeq{eq.VandermondeSumOfExponential} and the recursive definition of companion coefficients $\sgain_{k, \slambdaset}$ in \refeq{eq.CompanionCoefficientRecursion}, one can verify the result as follows:
\begin{align*}
\dot{\vfunc}_{k,\slambdaset}(t) \hspace{-6mm}&\hspace{6mm}=  (-1)^{\sorder - 1} \sum_{i=1}^{\sorder} \frac{\slambda_i \sgain_{k, \slambdaset_{\neg i}}}{\prod_{j \neq i} (\slambda_j - \slambda_i)} e^{\slambda_i t} 
\\
&= \slambda_l \underbrace{(-1)^{\sorder - 1} \sum_{i=1}^{\sorder} \frac{\sgain_{k, \slambdaset_{\neg i}}}{\prod_{j \neq i} (\slambda_j - \slambda_i)} e^{\slambda_i t}}_{\vfunc_{k, \slambdaset}(t)} \nonumber 
\\ 
&\hspace{10mm} +  (-1)^{\sorder - 1} \sum_{i=1}^{\sorder} \frac{(\slambda_i - \slambda_l) \sgain_{k, \slambdaset_{\neg i}}}{\prod_{j \neq i} (\slambda_j - \slambda_i)} e^{\slambda_i t}
\\
& = \slambda_l \vfunc_{k, \slambdaset}(t) +  (-1)^{\sorder - 2} \!\sum_{i\neq l} \frac{\sgain_{k, \slambdaset_{\neg i}}}{\prod_{\substack{j \neq i\\ j \neq l}} (\slambda_j - \slambda_i)} e^{\slambda_i t}
\\
& = \slambda_l \vfunc_{k, \slambdaset}(t) +  (-1)^{\sorder - 2} \!\sum_{i\neq l} \frac{-\slambda_l \sgain_{k, \slambdaset_{\neg i \neg l}} + \sgain_{k-1, \slambda_{\neg i \neg l}}}{\prod_{\substack{j \neq i\\ j \neq l}} (\slambda_j - \slambda_i)} e^{\slambda_i t} \!\!
\\
& = \slambda_l \vfunc_{k, \slambdaset}(t) - \slambda_l \vfunc_{k, \slambdaset_{\neg l}}(t) + \vfunc_{k-1, \slambdaset_{\neg l}} (t)  \qedhere
\end{align*}
\end{proof}

\subsection{Proof of \refprop{prop.VandermondeBasisNonnegativity}}
\label{app.VandermondeBasisNonnegativity}

\begin{proof}
Note that, by convention, $\vfunc_{k,\slambdaset}(t) = 0$ for all $k< 0$ and $k \geq |\slambdaset| = \sorder$. 
One also  has from \refprop{prop.VandermondeDynamics} that $\vfunc_{k, \slambdaset}(0) \geq 0$. 
Hence, we provide proof by induction:

$\bullet$ (Base Case) If $|\slambdaset| = 0$, then  the result trivially holds since $\vfunc_{k,\slambdaset}(t) = 0$ for all $k < 0$ and $k \geq |\slambdaset|=0$.

$\bullet$ (Induction) Otherwise, let $l \in \clist{1, \ldots, |\slambdaset|}$ and suppose that $\vfunc_{k, \slambdaset_{\neg l}}(t) \geq 0$  for any $k$. Then, observe from \refeq{eq.VandermondeDynamicsRecursion} that for any $k = 0, \ldots, \sorder-1 $ and $l = 1, \ldots, \sorder$
\begin{align*}
\dot{\vfunc}_{k,\slambdaset}(t) &= \slambda_l \vfunc_{k,\slambdaset}(t) \underbrace{- \, \slambda_l \vfunc_{k,\slambdaset_{\neg l}}(t)}_{\substack{\geq  0 \\ \text{by induction}\\ \text{and $\slambda_l < 0$}}} + \underbrace{\vfunc_{k-1, \slambdaset_{\neg l}}(t)}_{\substack{\geq 0\\ \text{by induction}}} \\
&\geq \slambda_l \vfunc_{k,\slambdaset}(t)
\end{align*}
Accordingly, one can conclude using the comparison lemma \cite{khalil_NonlinearSystems2001} that  $\vfunc_{k,\slambdaset}(t)\geq \vfunc_{k,\slambdaset}(0) e^{\slambda_l t} \geq 0 $ for all $t\geq 0$ since $\vfunc_{k,\slambdaset}(0)\geq 0$, which completes the proof. 
\end{proof}

\subsection{Proof of \refprop{prop.RelativeVandemondeBasisBound}}
\label{app.RelativeVandemondeBasisBound}

\begin{proof}
We shall provide proof by induction using the recursive Vandermonde dynamics in \refprop{prop.VandermondeDynamicsRecursion} and the companion coefficient recursion in \reflem{lem.CompanionCoefficientRecursion}.

$\bullet$ Base Case ($\sorder = 2$): Consider a second-order companion system with $\slambdaset = \plist{\slambda_i, \slambda_j}$. Note that there is only one relative bound relation between $\vfunc_{0, \slambdaset}$ and $\vfunc_{1, \slambdaset}$. 
Also observe that $\sgain_{0, \slambdaset} \!=\! \slambda_i \slambda_j$, $\sgain_{1, \slambdaset} \!=\! - (\slambda_i \!+\! \slambda_j)$,  $\sgain_{0, \slambdaset_{\neg i}} \!=\! - \slambda_j$ and $\sgain_{1, \slambdaset_{\neg i}} \!=\! 1$.
Hence, the difference $\sgain_{1, \slambdaset_{\neg i}} \vfunc_{0, \slambdaset} - \sgain_{0, \slambdaset_{\neg i}}\vfunc_{1, \slambdaset}$ becomes
\begin{align*}
\sgain_{1, \slambdaset_{\neg i}} \vfunc_{0, \slambdaset} - \sgain_{0, \slambdaset_{\neg i}}\vfunc_{1, \slambdaset} = \vfunc_{0, \slambdaset} + \slambda_j \vfunc_{1, \slambdaset}
\end{align*}
and its time rate of change can be obtained using \refeq{eq.VandermondeDynamics} as
\begin{align*}
&\frac{\diff }{\diff t} \plist{\sgain_{1, \slambdaset_{\neg i}} \vfunc_{0, \slambdaset} - \sgain_{0, \slambdaset_{\neg i}}\vfunc_{1, \slambdaset}} = \dot{\vfunc}_{0, \slambdaset} + \slambda_j\dot{\vfunc}_{1, \slambdaset} 
\\
& \hspace{15mm} =  \plist{\vfunc_{-1, \slambdaset} - \sgain_{0, \slambdaset} \vfunc_{1, \slambdaset}} + \slambda_j \plist{\vfunc_{0,\slambdaset} - \sgain_{1,\slambdaset} \vfunc_{1, \slambdaset}}   
\\
& \hspace{15mm} = - \slambda_i \slambda_j \vfunc_{1, \slambdaset} + \slambda_j(\vfunc_{0, \slambdaset} + (\slambda_i + \slambda_j) \vfunc_{1, \slambdaset})
\\
& \hspace{15mm} =  \slambda_j(\vfunc_{0, \slambdaset} +  \slambda_j \vfunc_{1, \slambdaset})
\\
& \hspace{15mm} = \slambda_j \plist{\sgain_{1, \slambdaset_{\neg i}} \vfunc_{0, \slambdaset} - \sgain_{0, \slambdaset_{\neg i}}\vfunc_{1, \slambdaset}}
\end{align*}
Therefore, since $\vfunc_{0, \slambdaset}(0) +  \slambda_j \vfunc_{1, \slambdaset}(0) = 1$ (see \refprop{prop.VandermondeDynamics}), we have
\begin{align*}
\sgain_{1, \slambdaset_{\neg i}} \vfunc_{0, \slambdaset} - \sgain_{0, \slambdaset_{\neg i}}\vfunc_{1, \slambdaset} = \vfunc_{0, \slambdaset} +  \slambda_j \vfunc_{1, \slambdaset} = e^{\slambda_j t} \geq 0.
\end{align*}

$\bullet$ Induction Step ($\sorder > 2$): Using the recursive Vandermonde dynamics in \refprop{prop.VandermondeDynamicsRecursion}, we observe that for any $j \neq i$, the time rate of change of the difference $\sgain_{k, \slambdaset_{\neg i}} \vfunc_{k-1, \slambdaset} - \sgain_{k-1, \slambdaset_{\neg i}}\vfunc_{k,\slambdaset} $ satisfies  
\begin{align*}
\hspace{10mm}& \hspace{-10mm}\frac{\diff}{\diff t} \plist{\sgain_{k, \slambdaset_{\neg i}} \vfunc_{k-1, \slambdaset} - \sgain_{k-1, \slambdaset_{\neg i}}\vfunc_{k,\slambdaset}} \\
& = \sgain_{k, \slambdaset_{\neg i}} \dot{\vfunc}_{k-1, \slambdaset} -  \sgain_{k-1, \slambdaset_{\neg i}} \dot{\vfunc}_{k,\slambdaset} \nonumber 
\\
&= \sgain_{k, \slambdaset_{\neg i}} \plist{\slambda_{j} \vfunc_{k-1, \slambdaset} - \slambda_{j} \vfunc_{k-1, \slambdaset_{\neg j}} + \vfunc_{k-2, \slambdaset_{\neg j}}} \nonumber 
\\
& \quad\quad - \sgain_{k-1, \slambdaset_{\neg i}} \plist{\slambda_{j} \vfunc_{k, \slambdaset} - \slambda_{j} \vfunc_{k, \slambdaset_{\neg j}} + \vfunc_{k-1, \slambdaset_{\neg j}}} \nonumber
\\
&= \slambda_{j}\plist{\sgain_{k, \slambdaset_{\neg i}} \vfunc_{k-1, \slambdaset} - \sgain_{k-1, \slambdaset_{\neg i}}\vfunc_{k,\slambdaset}}  + \Delta_{k} \nonumber
\end{align*}
where $\slambdaset_{\neg j} = \plist{\slambda_1, \ldots, \slambda_{j-1}, \slambda_{j+1}, \ldots, \slambda_{\sorder}}$, and the term $\Delta_k$ is defined as  
\begin{align}
\Delta_k & =  \slambda_{j} \sgain_{k-1,\slambdaset_{\neg i}} \vfunc_{k, \slambdaset_{\neg j}} + \sgain_{k, \slambdaset_{\neg i}} \vfunc_{k-2, \slambdaset_{\neg j}} \nonumber 
\\
& \hspace{15mm}  -\plist{\slambda_{j} \sgain_{k, \slambdaset_{\neg i}} + \sgain_{k-1, \slambdaset_{\neg_i}}} \vfunc_{k-1, \slambdaset_{\neg j}}. \nonumber
\end{align}
As an induction hypothesis,  $\vfunc_{k-2, \slambdaset_{\neg j}}\geq \frac{\sgain_{k-2, \slambdaset_{\neg i \neg j}}}{\sgain_{k-1, \slambdaset_{\neg i \neg j}}} \vfunc_{k-1, \slambdaset_{\neg j}}$ where $\sgain_{k-1, \slambdaset_{\neg i \neg j}} >0$  for any $k = 1, \ldots, \sorder -1$. 
Hence, we can find a lower bound on $\Delta_k$ as
\begin{align}
\Delta_k & \geq  \slambda_{j} \sgain_{k-1,\slambdaset_{\neg i}} \vfunc_{k, \slambdaset_{\neg j}} \nonumber 
\\
& \quad -\!\plist{\!\slambda_{j} \sgain_{k, \slambdaset_{\neg i}} \!+ \sgain_{k-1, \slambdaset_{\neg i}} \! - \sgain_{k, \slambdaset_{\neg i}}\frac{\sgain_{k-2, \slambdaset_{\neg i \neg j}}}{\sgain_{k-1, \slambdaset_{\neg i \neg j}}}\!} \vfunc_{k-1, \slambdaset_{\neg j}}.\nonumber
\end{align}

Similarly, applying another induction hypothesis of  $\vfunc_{k-1, \slambdaset_{\neg j}}\geq \frac{\sgain_{k-1, \slambdaset_{\neg i \neg j}}}{\sgain_{k, \slambdaset_{\neg i \neg j}}} \vfunc_{k, \slambdaset_{\neg j}}$, where $\sgain_{k, \slambdaset_{\neg i, \neg j}} > 0$ for all $k = 1, \ldots, \sorder -2$, we obtain for $k = 1, \ldots, \sorder -2$ that
\begin{align}
\Delta_k & \geq \vfunc_{k, \slambdaset_{\neg j}} \frac{\sgain_{k-1, \slambdaset_{\neg i}}}{\sgain_{k, \slambdaset_{\neg i \neg j}}}  \underbrace{\plist{\slambda_{j}  \sgain_{k, \slambdaset_{\neg i \neg j}}  -   \sgain_{k-1, \slambdaset_{\neg i \neg j}}}}_{= -\sgain_{k, \slambdaset_{\neg i}} \nonumber
}
\\
&\quad \quad +   \vfunc_{k, \slambdaset_{\neg j}} \frac{\sgain_{k, \slambdaset_{\neg i}}}{\sgain_{k, \slambdaset_{\neg i \neg j}}}    \underbrace{ \plist{- \slambda_{j}  \sgain_{k-1, \slambdaset_{\neg i \neg j}} +  \sgain_{k-2, \slambdaset_{\neg i \neg j}}}}_{= \sgain_{k-1, \slambdaset_{\neg i}}} \nonumber  
\\
& = 0 \nonumber
\end{align}
where the equality is due to the recursion of companion coefficients in \reflem{lem.CompanionCoefficientRecursion}.
For $k = \sorder -1$, the term $\Delta_{k}$ is still bounded below by zero as can be seen below
\begin{align}
\Delta_{\sorder -1} &\geq \sgain_{\sorder-2,\slambdaset_{\neg i}} \underbrace{\vfunc_{\sorder-1, \slambdaset_{\neg i \neg j }}}_{=0} \nonumber 
\\
& \hspace{-7mm} -\!\plist{\!\slambda_{j} \sgain_{\sorder-1, \slambdaset_{\neg i}} \! + \sgain_{\sorder -2, \slambdaset_{\neg i}} \!- \sgain_{\sorder-1, \slambdaset_{\neg i}}\frac{\sgain_{\sorder-3, \slambdaset_{\neg i \neg j}}}{\sgain_{\sorder-2, \slambdaset_{\neg i \neg j}}}\!} \vfunc_{\sorder-2, \slambdaset_{\neg j}}\nonumber
\\
&= - \slambda_j +\sgain_{k - 3, \slambdaset_{\neg i \neg j}} - \sgain_{k-2, \slambdaset_{\neg i}} = 0 \nonumber
\end{align}
which follows from the fact that $\sgain_{n-1, \slambdaset_{\neg i}} = \sgain_{n-2, \neg i \neg j} = 1$ and $\sgain_{k-2, \slambdaset_{\neg i}} = -\slambda_j + \sgain_{k-3, \slambdaset_{\neg i \neg j}}$.

Therefore, since $\Delta_k \geq 0$, the results follows from the comparison lemma \cite{khalil_NonlinearSystems2001} because 
\begin{align*}
&\frac{\diff}{\diff t} \plist{\sgain_{k, \slambdaset_{\neg i}} \vfunc_{k-1, \slambdaset} - \sgain_{k-1, \slambdaset_{\neg i}}\vfunc_{k,\slambdaset}} \nonumber \\
&\hspace{15mm} \geq \slambda_{j}\plist{\sgain_{k, \slambdaset_{\neg i}} \vfunc_{k-1, \slambdaset} - \sgain_{k-1, \slambdaset_{\neg i}}\vfunc_{k,\slambdaset}}
\end{align*}
and $\sgain_{k, \slambdaset_{\neg i}} \vfunc_{k-1, \slambdaset} (0) - \sgain_{k-1, \slambdaset_{\neg i}}\vfunc_{k,\slambdaset} (0) \geq 0$  for all $k = 1, \ldots, \sorder -1$ (see \refprop{prop.VandermondeDynamics}).
\end{proof}

\subsection{Proof of \refprop{prop.VandermondeBasisRatioLimit}}
\label{app.VandermondeBasisRatioLimit}

\begin{proof}
Using the explicit form of Vandermonde basis function as a weighted combination of exponential basis in \refeq{eq.VandermondeSumOfExponential}, the result can be obtained as
\begin{align*}
\lim_{t \rightarrow \infty} \frac{\vfunc_{k-1, \slambdaset}(t)}{\vfunc_{k, \slambdaset}(t)}
& = \lim_{t \rightarrow \infty} \frac{\sum_{i=1}^{\sorder} \frac{\sgain_{k-1, \slambdaset_{\neg i}}}{\prod_{j \neq i}(\slambda_j - \slambda_i)} e^{\slambda_i t}}{\sum_{i=1}^{\sorder} \frac{\sgain_{k, \slambdaset_{\neg i}}}{\prod_{j \neq i}(\slambda_j - \slambda_i)} e^{\slambda_i t}}
\\
& = \lim_{t \rightarrow \infty} \frac{\sum_{i=1}^{\sorder} \frac{\sgain_{k-1, \slambdaset_{\neg i}}}{\prod_{j \neq i}(\slambda_j - \slambda_i)} e^{(\slambda_i - \slambda_{\max}) t}}{\sum_{i=1}^{\sorder} \frac{\sgain_{k, \slambdaset_{\neg i}}}{\prod_{j \neq i}(\slambda_j - \slambda_i)} e^{(\slambda_i - \slambda_{\max}) t}}
\\
&  = \frac{\sgain_{k-1, \slambdaset_{\neg \max}}}{\sgain_{k, \slambdaset_{\neg \max}}}
\end{align*}
where the last equality follows from the monotonicity of $e^{\slambda t}$ and \mbox{$0 > \slambda_{\max} > \max(\slambdaset_{\neg \max})$}.
\end{proof}

\subsection{Proof of \refprop{prop.ZerothVandermondeBasisFunction}}
\label{app.ZerothVandermondeBasisFunction}

\begin{proof}
We have from \refprop{prop.VandermondeDynamics} that $\vfunc_{0,\slambdaset}(0) = 1$ and 
\begin{align*}
\dot{\vfunc}_{0,\slambdaset}(t) = -\sgain_{0, \slambdaset} \vfunc_{\sorder-1,\slambdaset}(t) = - \prod_{l=1}^{\sorder} (-\slambda_l) \vfunc_{\sorder-1,\slambdaset}(t) \leq 0  \!\!
\end{align*}
which follows  from $\vfunc_{\sorder-1,\slambdaset}(t) \geq 0$ (\refprop{prop.VandermondeBasisNonnegativity}) and $\sgain_{0,\slambdaset} = \prod_{l=1}^{\sorder} (-\slambda_l) >0$. Hence, the result holds.
\end{proof}

\subsection{Proof of \refprop{prop.SimplicialTrajectoryBound}}
\label{app.SimplicialTrajectoryBound}

\begin{proof}
By defining $\beta_\sorder = 0$ and  $\theta_{\sorder} (t)=0$ and observing the equation pattern, one can rewrite  the companion trajectory as
\begin{align*}
\sposn(t)  
&= \sum_{i=0}^{\sorder-1} \vect{y}_i(\sstate_0) \theta_i (t) =  \vect{y}_0(\sstate_0) \theta_0(t) +  \sum_{i=1}^{\sorder-1} \vect{y}_i(\sstate_0) \theta_i(t)  
\\
&= \plist{\frac{\beta_0}{\beta_0}\theta_0(t) - \frac{\beta_{1}}{\beta_0} \theta_1(t)\!\!} \frac{\beta_0}{\beta_0} \vect{y}_0(\sstate_0)  \nonumber 
\\
& \hspace{10mm}+ \frac{\beta_1}{\beta_0} \theta_1(t) \plist{\frac{\beta_0}{\beta_0} \vect{y}_0(\sstate_0) + \frac{\beta_0}{\beta_1} \vect{y}_1(\sstate_0)\!\!} \nonumber 
\\ 
&\hspace{20mm} +  \sum_{i=2}^{\sorder-1} \vect{y}_i(\sstate_0) \theta_i(t)
\\
&= \plist{\frac{\beta_0}{\beta_0} \theta_0(t) - \frac{\beta_{1}}{\beta_0} \theta_1(t)\!\!} \frac{\beta_0}{\beta_0} \vect{y}_0(\sstate_0) \nonumber 
\\
& \hspace{5mm} + \plist{\frac{\beta_1}{\beta_0} \theta_1(t) - \frac{\beta_2}{\beta_0} \theta_2(t)\!\!} \plist{\frac{\beta_0}{\beta_0}\vect{y}_0(\sstate_0) + \frac{\beta_0}{\beta_1} \vect{y}_1(\sstate_0)\!\!} \nonumber 
\\
&\hspace{10mm} + \frac{\beta_2}{\beta_0} \theta_2(t) \plist{\frac{\beta_0}{\beta_0}\vect{y}_0(\sstate_0) + \frac{\beta_0}{\beta_1} \vect{y}_1(\sstate_0) + \frac{\beta_0}{\beta_2} \vect{y}_2(\sstate_0)\!\!} \nonumber \!\!
\\ 
&\hspace{15mm} +  \sum_{i=3}^{\sorder-1} \vect{y}_i(\sstate_0) \theta_i(t)
\\ 
& = \sum_{i=0}^{\sorder -1} \plist{\frac{\beta_i}{\beta_0} \theta_i (t) - \frac{\beta_{i+1}}{\beta_0} \theta_{i+1}(t)\!\!} \sum_{j=0}^{i} \frac{\beta_0}{\beta_j} \vect{y}_j(\sstate_0).
\end{align*}

Due to the relative ordering of the basis functions, we have $\frac{\beta_i}{\beta_0} \theta_i (t) - \frac{\beta_{i+1}}{\beta_0} \theta_{i+1}(t) \geq 0$.
Moreover, one can observe from the boundedness property that $\sum\limits_{i=0}^{\sorder -1} \plist{\frac{\beta_i}{\beta_0} \theta_i (t) - \frac{\beta_{i+1}}{\beta_0} \theta_{i+1}(t)\!} = \theta_0(t) - \frac{\beta_{\sorder}}{\beta_0}\theta_{\sorder}(t) = \theta_0(t) \leq 1$.
Hence, we conclude that the companion trajectory $\sposn(t)$ is a convex combination of $\vect{y}_0(\sstate_0), \vect{y}_0(\sstate_0) + \frac{\beta_1}{\beta_0} \vect{y}_1(\sstate_0) , \ldots, \sum_{j=1}^{\sorder-1} \frac{\beta_j}{\beta_0} \vect{y}_j(\sstate_0)$ with the origin (whose weight is $1 - \theta_0(t)$).
Thus, the result holds.
\end{proof}

\subsection{Proof of \refprop{prop.LyapunovEllipsoid}}
\label{app.LyapunovEllipsoid}

\begin{proof}
Let $\sstate(t)$ denote the state-space solution trajectory of the companion dynamics $\dot{\sstate} = (\cmat_{\slambdaset} \otimes \mat{I}_{\sdim \times \sdim}) \sstate$ starting at $t = 0$ from the initial state $\sstate_0$.
Consider the quadratic Lyapunov function $V_{\mat{P}}(\sstate) = \tr{\sstate} \mat{P} \sstate$ that is parameterized by the positive definite Lyapunov matrix $\mat{P}$ satisfying \refeq{eq.CompanionLyapunovEquation}.
Since the time rate of change of the Lyapunov function is nonincreasing, i.e.,
\begin{align*}
\dot{V}_{\mat{P}}(\sstate) &= \tr{\sstate}(\tr{(\cmat_{\slambdaset} \otimes \mat{I}_{\sdim \times \sdim})} \mat{P} + \mat{P} (\cmat_{\slambdaset} \otimes \mat{I}_{\sdim \times \sdim})) \sstate \\
& = 	- \norm{\mat{D} \sstate}^2 \leq 0
\end{align*}
the state-space trajectory $\sstate(t)$ of the companion system is contained in the Lyapunov ellipsoid $\elp(\mat{0}, \mat{P}^{-1}, \norm{\sstate_0}_{\mat{P}})$  for all $t \geq 0$.
Hence, since the position variable $\sposn$ and state-space variable $\sstate$ of the companion system are related to each other by the orthogonal transformation  $\sposn = \mat{I}_{\sdim \times \sorder \sdim} \sstate$, the orthogonal projection of the Lyapunov ellipsoid $\elp(\mat{0}, \mat{P}^{-1}, \norm{\sstate_0}_{\mat{P}})$ onto the column space of $\mat{I}_{\sdim \times \sorder \sdim}$  is another ellipsoid  $\elp(\mat{0}, \tr{\mat{I}_{\sorder \sdim \times \sdim}}\mat{P}^{-1} \mat{I}_{\sorder \sdim \times \sdim}, \norm{\sstate_0}_{\mat{P}})$ that contains the position trajectory $\sposn(t)$ for all $t \geq 0$ (see \reflem{lem.EllipsoidOrthogonalProjection}). 
\end{proof}

\vspace{-3mm}

\section{Vandermonde Determinant Recursion}

Vandermonde matrices enjoy determinant recursion \cite{neagoe_SPL1996}. 

\vspace{-1mm}

\begin{lemma}\label{lem.VandermondeDeterminantRecursion}
\emph{(Vandermonde Determinant Recursion)}
For any $\slambdaset = [\slambda_1, \ldots, \slambda_\sorder] \in \C^{\sorder}$, the Vandermonde matrix determinant can be recursively determined as 
\begin{align}
\det \vmat_{\slambdaset} = (-1)^{l-1}\prod_{k \neq l} (\slambda_k - \slambda_l) \det \vmat_{\slambdaset_{\neg l}}
\end{align}
with the base case $\det \vmat_{\slambdaset} \!= 1$ for $\sorder = 1$, where  \mbox{$l \!\in \!\{ 1, \ldots, \sorder\}$} and $\slambdaset_{\neg l} = [\slambda_1, \ldots, \slambda_{l-1}, \slambda_{l+1}, \slambda_{\sorder}]$. 
\end{lemma}
\begin{proof}
By applying a series of elementary matrix operations (e.g., starting with subtracting $\slambda_l$ times $(\sorder-1)^{th}$ row from the $\sorder^{th}$ row), one can verify the result as in \refeq{eq.VandermondeDeterminantRecursion} in \reftab{tab.VandermondeDeterminantRecursion}. 
\end{proof}

\begin{table*}[h]
\caption{Vandermonde Determinant Recursion}
\label{tab.VandermondeDeterminantRecursion}
\vspace{-5mm}
\centering
\hrulefill
\begin{subequations} \label{eq.VandermondeDeterminantRecursion}
\begin{align}
\det \vmat_{\slambdaset} &= 
\scalebox{0.8}{$ 
\left |
\begin{matrix}
1 & \ldots & 1  & 1& 1 & \ldots & 1 \\
\slambda_1 & \ldots  &\slambda_{l-1} & \slambda_l & \slambda_{l+1}& \ldots & \slambda_\sorder \\
\slambda_1^2 & \ldots  &\slambda_{l-1}^2 & \slambda_l^2 & \slambda_{l+1}^2& \ldots & \slambda_\sorder^2 \\
\vdots  & \ddots  &\vdots & \vdots & \vdots& \ddots  & \vdots \\
\slambda_1^{\sorder-1} & \ldots  &\slambda_{l-1}^{\sorder-1} & \slambda_l^{\sorder-1} & \slambda_{l+1}^{\sorder-1}& \ldots & \slambda_\sorder^{\sorder-1}  
\end{matrix}
\right | 
$}
= 
\scalebox{0.8}{$ 
\left |
\begin{matrix}
1 & \ldots & 1  & 1& 1 & \ldots & 1 \\
(\slambda_1 - \slambda_l) & \ldots  &(\slambda_{l-1}-\slambda_l) & 0 & (\slambda_{l+1}-\slambda_l)& \ldots & (\slambda_\sorder-\slambda_l) \\
\slambda_1(\slambda_1 - \slambda_l) & \ldots  &\slambda_{l-1}(\slambda_{l-1} - \slambda_l) & 0 & \slambda_{l+1} (\slambda_{l+1} - \slambda_l)& \ldots & \slambda_\sorder(\slambda_\sorder - \slambda_l) \\
\vdots  & \ddots  &\vdots & \vdots & \vdots& \ddots  & \vdots \\
\slambda_1^{\sorder-2}(\slambda_1 - \slambda_l) & \ldots  &\slambda_{l-1}^{\sorder-2}(\slambda_{l-1} - \slambda_l) & 0 & \slambda_{l+1}^{\sorder-2}(\slambda_{1+1} - \slambda_l)& \ldots & \slambda_\sorder^{\sorder-2} (\slambda_\sorder - \slambda_l)  
\end{matrix}
\right |
$} 
\\
& = 
\scalebox{0.8}{$ 
(-1)^{l-1}
\left |
\begin{matrix}
(\slambda_1 - \slambda_l) & \ldots  &(\slambda_{l-1}-\slambda_l) & (\slambda_{l+1}-\slambda_l)& \ldots & (\slambda_\sorder-\slambda_l) \\
\slambda_1(\slambda_1 - \slambda_l) & \ldots  &\slambda_{l-1}(\slambda_{l-1} - \slambda_l) & \slambda_{l+1} (\slambda_{l+1} - \slambda_l)& \ldots & \slambda_\sorder(\slambda_\sorder - \slambda_l) \\
\vdots  & \ddots  &\vdots &  \vdots& \ddots  & \vdots \\
\slambda_1^{\sorder-2}(\slambda_1 - \slambda_l) & \ldots  &\slambda_{l-1}^{\sorder-2}(\slambda_{l-1} - \slambda_l)  & \slambda_{l+1}^{\sorder-2}(\slambda_{1+1} - \slambda_l)& \ldots & \slambda_\sorder^{\sorder-2} (\slambda_\sorder - \slambda_l)  
\end{matrix}
\right |
$}
\\
&= 
\scalebox{0.8}{$ 
(-1)^{l-1} \prod\limits_{k \neq l} (\slambda_k - \slambda_l)
\left |
\begin{matrix}
1 & \ldots  &1 & 1&  \ldots & 1 \\
\slambda_1 & \ldots  &\slambda_{l-1} & \slambda_{l+1} & \ldots & \slambda_\sorder \\
\vdots  & \ddots  &\vdots &  \vdots& \ddots  & \vdots \\
\slambda_1^{\sorder-2} & \ldots  &\slambda_{l-1}^{\sorder-2} & \slambda_{l+1}^{\sorder-2}& \ldots & \slambda_\sorder^{\sorder-2}   
\end{matrix}
\right |
$}
= (-1)^{l-1} \prod\limits_{k \neq l} (\slambda_k - \slambda_l) \det \vmat_{\slambdaset_{\neg l}}
\end{align}
\end{subequations}
\hrulefill
\end{table*}

By successively applying \reflem{lem.VandermondeDeterminantRecursion}, one can conclude: 
\begin{corollary}
For any $\slambdaset = [\slambda_1, \ldots, \slambda_\sorder] \in \C^{\sorder}$, the Vandermonde matrix determinant is given by
\begin{align}
\det \vmat_{\slambdaset} = \prod_{1\leq k \leq l \leq \sorder} (\slambda_l - \slambda_k).
\end{align}
\end{corollary}

\vspace{-3mm}

\section{Companion Coefficient Recursion}
\label{app.CompanionCoefficientRecursion}

An essential property of companion coefficients in \refeq{eq.ControlGain} is their recursive nature that plays significant role in Vandermonde inverse in \refeq{eq.VandermondeInverse} and Vandermonde simplexes in  \refeq{eq.VandermondeSimplex}.

\vspace{-1mm}

\begin{lemma} \label{lem.CompanionCoefficientRecursion}
\emph{(Companion Coefficient Recursion)}
For any given  $\slambdaset = \blist{\slambda_1, \ldots, \slambda_{\sorder}} \in \C^{\sorder}$, the companion coefficients $\sgain_{k, \slambdaset}$ can be recursively determined for $0 \leq k \leq  \sorder-1$ as
\begin{align}\label{eq.CompanionCoefficientRecursion}
\sgain_{k, \slambdaset} = - \slambda_i \sgain_{k, \slambdaset_{\neg i}} + \sgain_{k-1, \slambdaset_{\neg i}}
\end{align} 
with the base case 
\begin{align}\label{eq.CompanionCoefficientBase}
\sgain_{k, \slambdaset} =  \left \{ \begin{array}{cl} 1, & \text{if } k = \sorder,\\
0, & \text{if } k \not \in [0,  \sorder],   \end{array} \right.
\end{align}
where $i \in\! [1, \ldots, \sorder]$ and  $\slambdaset_{\neg i} = \! [\slambda_1, \ldots, \slambda_{i-1}, \slambda_{i+1}, \ldots, \slambda_{\sorder}]$.
\end{lemma}
\begin{proof}
The base cases follow from the standard conversion that summation and multiplication over the empty set are zero and one, respectively.
The result can be verified using \refeq{eq.ControlGain} as
\begin{align*}
\sgain_{k, \slambdaset} &= \sum_{\substack{I \subseteq \clist{1, \ldots, |\slambdaset|} \\ |I| = |\slambdaset|  - k}} \prod_{i \in I} (- \slambda_i),
\\
&= -\slambda_i \plist{\sum_{\substack{J \subseteq \clist{1, \ldots,i-1, i+1, \ldots  |\slambdaset|} \\ |J| = |\slambdaset| - 1  - k}} \prod_{j \in J} (- \slambda_j)} \nonumber 
\\ 
& \hspace{15mm}+    \sum_{\substack{J \subseteq \clist{1, \ldots,i-1, i+1, \ldots  |\slambdaset|} \\ |J| = |\slambdaset|  - k}} \prod_{j \in J} (- \slambda_j), 
\\&=  -\slambda_i \sgain_{k, \slambdaset_{\neg i}} + \sgain_{k -1, \slambdaset_{\neg i}},
\end{align*}
which competes the proof.
\end{proof}

\noindent Since $\sgain_{-1, \slambdaset} = 0$ and $\sgain_{|\slambdaset|, \slambdaset} = 1$, the companion coefficient recursions for $\sgain_{0, \slambdaset}$ and $\sgain_{|\slambdaset|-1, \slambdaset}$ can be simplified as
\begin{align*}
\sgain_{0, \slambdaset} &= -\slambda_i \sgain_{0,\slambdaset_{\neg i}} = \prod_{i=1}^{|\slambdaset|} (-\slambda_i), 
\\
\sgain_{|\slambdaset|-1, \slambdaset} &= -\slambda_i + \sgain_{|\slambdaset|-2, \slambdaset_{\neg i}} = \sum_{i=1}^{|\slambdaset|} (-\slambda_i).
\end{align*}

\section{\mbox{Affine Transformation \small{and} Orthogonal Projection} of Ellipsoids}

The family of ellipsoids is closed under affine transformations and orthogonal projections, that is to say, an affine transformation or orthogonal transformation of an ellipsoid yields another ellipsoid. 
We provide below explicit formulas for performing these operations on ellipsoids as a reference.

Let $\elp(\elpctr, \elpmat, \elprad)$ denote the ellipsoid centered at $\elpctr \in \R^{n}$ and associated with a positive semidefinite matrix $\elpmat \in \PSDM^n$ and a nonnegative scaling factor $\elprad \in \R_{\geq 0}$ that is defined as\footnote{%
If $\elpmat$ is positive definite and so invertible, then the ellipsoid $\elp(\elpctr, \elpmat, \elprad)$ can be equivalently expressed as
\begin{align}
\elp(\elpctr, \elpmat, \elprad) = \clist{\vect{x} \in \R^{n} \Big | \tr{(\vect{x} - \elpctr)} \elpmat^{-1} (\vect{x} - \elpctr) \leq \elprad^2}. \nonumber
\end{align}
}
\begin{align}\label{eq.Ellipsoid}
\elp(\elpctr, \elpmat, \elprad) = \clist{\elpctr + \elpmat^{\frac{1}{2}} \vect{u} \big | \vect{u} \in \R^{n},  \norm{\vect{u}} \leq \elprad}, 
\end{align}
where $\elpmat^{\frac{1}{2}}$ is a square root\reffn{fn.MatrixSquareRoot} of $\elpmat$ that satisfies $ \elpmat^{\frac{1}{2}} \tr{(\elpmat^{\frac{1}{2}})\!} = \elpmat$, and $\norm{.}$ denotes the standard Euclidean norm.

\subsection{Affine Transformations of Ellipsoids}

An affine transformation applies a linear mapping on an ellipsoid followed by a translation, and so has a simple form.
\begin{lemma} \label{lem.EllipsoidAffineTransformation}
\emph{(Affine Transformation of Ellipsoids)}
For any $\mat{A} \in \R^{m \times n}$ and $\vect{b} \in \R^{m}$, the affine transformation $\vect{x} \rightarrow \mat{A} \vect{x} + \vect{b}$ maps an ellipsoid $\elp(\elpctr, \elpmat, \elprad)$ in $\R^n$  to another ellipsoid $ \elp(\mat{A} \elpctr + \vect{b}, \mat{A} \elpmat \tr{\mat{A}}, \elprad) $ in $\R^{m}$, i.e.,  
\begin{align}
\!\!\clist{ \mat{A}\vect{x} + \vect{b} \big| \vect{x} \in \!\elp(\elpctr, \elpmat, \elprad)} = \elp(\mat{A} \elpctr  + \vect{b}, \mat{A} \elpmat \tr{\mat{A}}\!, \elprad). \!\!\!
\end{align} 
\end{lemma}
\begin{proof}
The result trivially holds if $\mat{A} \elpmat^{\frac{1}{2}} = \mat{0}_{m \times n}$. 
Hence, we consider below the case $\mat{A} \elpmat^{\frac{1}{2}} \neq \mat{0}_{m \times n}$.

We first start with the special case of orthogonal transformations of the unit Euclidean ball, that is to say, $\mat{A}$ is orthogonal, $\vect{b} = \mat{0}$, and $\elpmat = \mat{I}_{n \times n}$ and $\elprad = 1$.  
Any orthogonal matrix $\mat{Q} \in \R^{n \times m}$ with $\tr{\mat{Q}} \mat{Q} = \mat{I}_{m \times m}$ satisfies 
\begin{align}\label{eq.OrthogonalTransformationEuclideanBall}
\clist{\tr{\mat{Q}} \vect{v} \big | \vect{v} \in \R^{n}, \norm{\vect{v}} \leq 1} = \clist{\vect{u} \in \R^{m} \big| \norm{\vect{u}} \leq 1},
\end{align}
because for any $\vect{v} \in \R^{n}$ one has  $ \norm{\tr{\mat{Q}} \vect{v}} \leq \lambda_{\max}(\mat{Q}) \norm{\vect{v}}^2 \leq \norm{\vect{v}}^2 $ which implies 
\begin{align}
\clist{\tr{\mat{Q}} \vect{v} | \vect{v} \in \R^{n}, \norm{\vect{v}} \leq 1} \subseteq \clist{\vect{u} \in \R^{m} | \norm{\vect{u}} \leq 1},
\end{align} 
and, for any $\vect{u} \in \R^{m}$ and $\vect{v} = \mat{Q} \vect{u}$ one has $\tr{\vect{v}}  \vect{v} = \tr{\vect{u}} \tr{\mat{Q}} \mat{Q} \vect{u}= \tr{\vect{u}}  \vect{u}$ which implies 
\begin{align}
\clist{\tr{\mat{Q}} \vect{v} | \vect{v} \in \R^{n}, \norm{\vect{v}} \leq 1} \supseteq \clist{\vect{u} \in \R^{m} | \norm{\vect{u}} \leq 1}.
\end{align}

In general, let $\mat{A} \elpmat^{\frac{1}{2}} = \mat{U} \mat{\Lambda} \tr{\mat{V}}$ be the singular value decomposition of $\mat{A} \elpmat^{\frac{1}{2}}$ associated with orthogonal matrices $\mat{U} \in \R^{m \times r}$ and $\mat{V} \in \R^{n \times r}$ (i.e., $\tr{\mat{U}}\mat{U} = \tr{\mat{V}} \mat{V} = \mat{I}_{r \times r}$) and nonsingular diagonal matrix $\mat{\Lambda} \in \R^{r \times r}$.
Hence, we have $\mat{A} \elpmat \tr{\mat{A}} = \mat{A} \elpmat^{\frac{1}{2}} \tr{\plist{\mat{A} \elpmat^{\frac{1}{2}}}} = \mat{U} \mat{\Lambda}^{2} \tr{\mat{U}}$.
Therefore, using the orthogonal transformation of the unit Euclidean ball in \refeq{eq.OrthogonalTransformationEuclideanBall},
we can determine the affine transformation of  $\elp(\elpctr, \elpmat, \elprad^2)$ as
\begin{align*}
&\clist{ \mat{A}\vect{x} + \vect{b} \big| \vect{x} \in \!\elp(\elpctr, \elpmat, \elprad)} \nonumber \\
& \hspace{6mm} = \clist{\mat{A} \elpmat^{\frac{1}{2}} \vect{v} + \mat{A} \elpctr + \vect{b} \big | \vect{v} \in \R^{n} , \norm{\vect{v}} \leq \elprad },
\\
&  \hspace{6mm} = \clist{\mat{U} \mat{\Lambda} \tr{\mat{V}} \vect{v} + \mat{A} \elpctr + \vect{b} \big | \vect{v} \in \R^{n} , \norm{\vect{v}} \leq \elprad },
\\
& \hspace{6mm} = \clist{\mat{U} \mat{\Lambda} \vect{w}  +  \mat{A} \elpctr + \vect{b} \big | \vect{w} \in \R^{r} , \norm{\vect{w}} \leq 1\elprad  },
\\
& \hspace{6mm} = \clist{\mat{U} \mat{\Lambda} \tr{\mat{U}} \vect{u}  +  \mat{A} \elpctr + \vect{b} \big | \vect{u} \in \R^{m} , \norm{\vect{u}} \leq \elprad  },
\\
& \hspace{6mm} = \clist{\plist{ \mat{A} \elpmat \tr{\mat{A}}}^{\frac{1}{2}} \vect{u} +  \mat{A} \elpctr + \vect{b} | \vect{u} \in \R^{m}, \norm{\vect{u}} \leq \elprad}, \!\!
\\ 
& \hspace{6mm} = \elp(\mat{A} \elpctr + \vect{b}, \mat{A} \elpmat \tr{\mat{A}}, \elprad),
\end{align*}
which completes the proof.
\end{proof}

\subsection{Orthogonal Projection of Ellipsoids}

In this part, we briefly present the explicit forms of the orthogonal projections of points and ellipsoids onto affine subspaces \cite{karl_verghese_willsky_CVGIP1994, pope_AlgorithmEllipsoid2008}. 

Let $A(\mat{Q}, \vect{p})$ denote the affine subspace that is obtained by shifting the column space of  an orthogonal matrix $\mat{Q} \in \R^{n \times m}$ (i.e., $\mat{Q}^T \mat{Q} = \mat{I}_{m \times m}$) by a translation of $\vect{p} \in \R^{n}$, 
\begin{align}
A(\mat{Q}, \vect{p}) = \clist{\mat{Q} \vect{y} +  \vect{p}  \; \big | \; \vect{y} \in \R^{m}}.
\end{align}  
In other words,  $A(\mat{Q}, \vect{p})$ is the affine orthogonal transformation of $\R^{m}$ by the map $\vect{x} \mapsto \mat{Q} \vect{x} + \vect{p}$.
We also denote by $\proj_{A}(\vect{x})$ the (metric) projection of a point $\vect{x} \in \R^{n}$ onto a (closed) set $A \subseteq \R^{n}$ that is defined as 
\begin{align}
\proj_{A}(\vect{x}) := \argmin_{\vect{a} \in A} \norm{\vect{a} - \vect{x}}.
\end{align}

\begin{lemma} \label{lem.PointOrthogonalProjection}
\emph{(Orthogonal Projection of Points)}
For any given orthogonal matrix $\mat{Q} \in \R^{n \times m}$ and $\vect{p} \in \R^{n}$, the orthogonal projection $\proj_{A(\mat{Q}, \vect{p})}(\vect{x})$ of a point $\vect{x} \in \R^{n}$ is given by
\begin{align}
\proj_{A(\mat{Q}, \vect{p})}(\vect{x}) = \mat{Q} \tr{\mat{Q}} (\vect{x} - \vect{p}) + \vect{p}
\end{align}
which is the affine orthogonal transformation of $\tr{\mat{Q}} (\vect{x} - \vect{p})$ via $\vect{y} \mapsto \mat{Q} \vect{y} + \vect{p}$ 
\end{lemma}
\begin{proof}
The metric projection onto an affine subspace can be rewritten as a convex least squares problem as 
\begin{align}
\proj_{A(\mat{Q}, \vect{p})}(\vect{x}) & = \argmin_{\vect{a} \in A(\mat{Q}, \vect{p})} \norm{\vect{a} - \vect{x}}\\
&  = \mat{Q} \plist{\argmin_{\vect{y} \in \R^{m}} \norm{\mat{Q} \vect{y} + \vect{p} - \vect{x}}^2} + \vect{p} 
\end{align} 
which is globally minimized at $\vect{y} = \tr{\mat{Q}} (\vect{x} - \vect{p})$ since 
\begin{align}
\nabla_{\vect{y}} \norm{\mat{Q} \vect{y} + \vect{p} - \vect{x}}^2 = \tr{\mat{Q}}(\mat{Q}\vect{y} + \vect{p} - \vect{x}) = 0,
\end{align}
where $\tr{\mat{Q}} \mat{Q} = \mat{I}_{m \times m}$.  Thus, the result holds.
\end{proof}

\begin{lemma}\label{lem.EllipsoidOrthogonalProjection} 
\emph{(Orthogonal Projection of Ellipsoids)}
For any given orthogonal matrix $\mat{Q} \in \R^{n \times m}$ and $\vect{p} \in \R^{n}$, the orthogonal projection $\proj_{A(\mat{Q}, \vect{p})}(\elp(\elpctr, \elpmat, \elprad))$ of an ellipsoid $\elp(\elpctr, \elpmat, \elprad)$ onto an affine subspace $A(\mat{Q}, \vect{p})$ is another ellipsoid $\elp(\mat{Q}\mat{Q}^T (\elpctr - \vect{p}) + \vect{p}, \mat{Q}\mat{Q}^T \elpmat \mat{Q}\mat{Q}^T, \elprad))$ which is the affine orthogonal transformation of a lower-dimensional ellipsoid $\elp(\mat{Q}^T (\elpctr - \vect{p}), \mat{Q}^T \elpmat \mat{Q}, \elprad)$ via the mapping $\vect{x} \mapsto \mat{Q} \vect{x} + \vect{p}$, i.e.,
\begin{align*}
\Pi_{A(\mat{Q}, \vect{p})}\!(\elp(\elpctr, \elpmat, \elprad)) \!& = \elp(\mat{Q}\mat{Q}^T (\elpctr\! -\! \vect{p}) \!+\! \vect{p}, \mat{Q}\mat{Q}^T \elpmat \mat{Q}\mat{Q}^T, \elprad)\!)
\\
& \hspace{-5mm} = \clist{\mat{Q} \vect{x} + \vect{p} \big | \vect{x} \in \elp(\mat{Q}^T \!(\elpctr\! -\! \vect{p}), \mat{Q}^T \elpmat \mat{Q}, \elprad)}. \! 
\end{align*}
\end{lemma}
\begin{proof}
The orthogonal project $\proj_{A(\mat{Q}, \vect{p})}(\vect{x})$ of a point $\vect{x} \in \R^{n}$ onto the affine subspace $A(\mat{Q}, \vect{p})$ is  given by (\reflem{lem.PointOrthogonalProjection}) 
\begin{align}
\proj_{A(\mat{Q}, \vect{p})}(\vect{x}) = \mat{Q} \mat{Q}^{T}(\vect{x} - \vect{p}) + \vect{p}.
\end{align} 
Hence, the projection of an ellipsoid point $\vect{x}= \elpctr +  \elpmat^{\frac{1}{2}} \vect{u}  \in \elp(\elpctr, \elpmat, \elprad) = \clist{\elpctr + \elpmat^{\frac{1}{2}} \vect{u} \big | \vect{u} \in \R^{n}, \norm{\vect{u}} \leq \elprad}$ onto $A(\mat{Q}, \vect{p})$ is given by
\begin{align*}
\mat{Q}\mat{Q}^T\!(\elpctr \!+\! \elpmat^{\frac{1}{2}} \vect{u} \!-\! \vect{p}) \!+\! \vect{p} = \mat{Q}\mat{Q}^T\!(\elpctr \!-\! \vect{p})  \!+\! \vect{p} \!+\! \mat{Q} \mat{Q}^T \elpmat^{\frac{1}{2}} \vect{u}.  
\end{align*}
Therefore, one can observe that the projection of an  ellipsoid  is another ellipsoid as follows 
\begin{align*}
\Pi_{A(\mat{Q}, \vect{p})}(\elp(\elpctr, \elpmat, \elprad)) &  = \clist{\Pi_{A(\vect{p}, \mat{Q})}(\vect{x}) | \vect{x} \in \elp(\elpctr, \elpmat, \elprad)} 
\\
& \hspace{-10mm} = \clist{\mat{Q}\mat{Q}^T(\elpctr +  \elpmat^{\frac{1}{2}} \vect{u} - \vect{p}) | \norm{\vect{u}} \leq \elprad}
\\
& \hspace{-10mm} = \clist{\mat{Q}\mat{Q}^T(\elpctr - \vect{p})  + \vect{p} + \mat{Q}\mat{Q}^T \elpmat^{\frac{1}{2}} \vect{u} | \norm{\vect{u}} \leq \elprad}
\\
& \hspace{-10mm} = \elp(\mat{Q}\mat{Q}^T (\elpctr - \vect{p}) + \vect{p}, \mat{Q}\mat{Q}^T \elpmat \mat{Q}\mat{Q}^T, \elprad))
\end{align*}
because $\mat{Q}\mat{Q}^T \elpmat^{\frac{1}{2}} \tr{(\mat{Q}\mat{Q}^T \elpmat^{\frac{1}{2}})} = \mat{Q}\mat{Q}^T \elpmat \mat{Q}\mat{Q}^T$.
Moreover, it follows from \reflem{lem.EllipsoidAffineTransformation} that the orthogonal projection of an ellipsoid can be rewritten as an orthogonal transformation of another ellipsoid as 
\begin{align*}
&\proj_{A(\vect{p}, \mat{Q})}(\elp(\elpctr, \elpmat, \elprad))\! = \!\clist{\mat{Q} \vect{x} \!+\! \vect{p} \big | \vect{x} \!\in\!  \elp(\mat{Q}^T\! (\elpctr\! - \!\vect{p}), \mat{Q}^T \elpmat \mat{Q}, \elprad)\!}
\end{align*}
which completes the proof.
\end{proof}
\noindent Note that $\mat{P} = \mat{Q}\mat{Q}^T$ is known as the orthogonal projection matrix (i.e., $\mat{P}^2 = \mat{P}$ and $\mat{P}^T = \mat{P}$) defining the orthogonal projection of ellipsoids on the affine subspace $A(\mat{Q}, \vect{p})$.

\section{Projected Lyapunov Ellipsoids for Linear Systems}

Invariant sublevel sets of Lyapunov functions are widely used for constrained control and optimization of dynamical systems \cite{blanchini_Automatica1999}.
In particular, quadratic Lyapunov functions of stable linear time-invariant systems offer analytic ellipsoidal bounds on their state-space trajectories.  
In this part, we present orthogonal projections of Lyapunov ellipsoids to handle partial state constraints; for example,  spatial safety constraints expressed in terms of system position, and control input constraints expressed in terms of system velocity and acceleration.  

Consider an exponentially stable linear-time invariant system whose state vector $\vect{y} \in \R^{n}$ evolves based on
\begin{align*}
\dot{\vect{y}} = \mat{A} \vect{y}
\end{align*}
where  $\mat{A}$ is a stability (Hurwitz) matrix whose eigenvalues have strictly negative real parts. 
To express future system motion, starting from any initial state $\vect{y}(0) \in \R^{n}$, one can construct a quadratic Lyapunov function $V_{\mat{P}}(\vect{y}) = \vect{y}^T \mat{P} \vect{y}$ parametrized by a positive definite symmetric matrix $\mat{P} \in \PDM^{n}$ that uniquely solves the Lyapunov equation
\begin{align*}
\mat{A}^T \mat{P} + \mat{P} \mat{A} + \mat{C}^T \mat{C} = 0
\end{align*}
for some matrix $\mat{C} \in \R^{r \times n}$  such that $(\mat{A}, \mat{C})$ is observable \cite{khalil_NonlinearSystems2001}. 
Since the time rate of change of the Lyapunov function is nonincreasing, i.e.,
\begin{align*}
\dot{V}_{\mat{P}}(\vect{y}) = \vect{y}^T (\mat{A}^T \mat{P} + \mat{P} \mat{A}) \vect{y} = -\norm{\mat{C} \vect{y}}^2 \leq 0,
\end{align*}
the state-space trajectory $\vect{y}(t)$ of the system is contained in the Lyapunov ellipsoid $\elp(\mat{0}, \mat{P}^{-1}, \norm{\vect{y}(0)}_{\mat{P}})$, i.e.,
\begin{align*}
\vect{y}(t)  & \in \elp(\mat{0}, \mat{P}^{-1}, \norm{\vect{y}(0)}_{\mat{P}})  \quad \forall t \geq 0
\end{align*}
where $\elp(\elpctr, \elpmat, \elprad)$ denotes the ellipsoid centered at $\elpctr \in \R^{n}$ and associated with a positive semidefinite matrix $\elpmat \in \PSDM^{n}$ and a nonnegative scaling factor $\elprad \in \R_{\geq 0}$ that is defined as in \refeq{eq.Ellipsoid}, 
and $\norm{\vect{x}}_{\mat{W}} := \sqrt{\tr{\vect{x}} \mat{W} \vect{x}}$ is the weighted Euclidean norm associated with a positive definite matrix $\mat{W} \in \PDM^{n}$ and  $\norm{.}$ denotes the standard Euclidean norm.

In general, various system constraints are specified in terms of different subsets of system state variables e.g., position constraints for safe motion and velocity/acceleration constraints for control inputs.
Instead of converting such system constraints into high-dimensional state-space constraints, one can handle them in lower-dimension subspaces by using orthogonal projection of system trajectory and Lyapunov ellipsoids. 

\begin{proposition} \label{prop.ProjectedLyapunovEllipsoid}
\emph{(Projected Lyapunov Ellipsoids)} 
Given a quadratic  Lyapunov  function $V_{\mat{P}}(\vect{y}) = \tr{\vect{y}} \mat{P} \vect{y} $ for a stable linear-time invariant system $\dot{\vect{y}} = \mat{A} \vect{y}$, the orthogonal projection of the state-space trajectory $\vect{y} (t)$ of the system, starting from any initial state $\vect{y}(0) \in \R^{n}$, onto the column space of an orthogonal matrix $\mat{Q} \in \R^{n \times m}$  can be explicitly bounded by the orthogonal projection of the Lyapunov ellipsoid $\elp(\mat{0}, \mat{P}^{-1}, \norm{\vect{y}(0)}_{\mat{P}})$  onto the same subspace as
\begin{align*}
\tr{\mat{Q}} \vect{y}(t) \in \elp(\mat{0}, \tr{\mat{Q}} \mat{P}^{-1} \mat{Q}, \norm{\vect{y}(0)}_{\mat{P}}) \quad \forall t \geq 0.
\end{align*}
\end{proposition}
\begin{proof}
Denote by $\colspace(\mat{Q})$ the column space of $\mat{Q}$ that is the set of all linear combinations of its column vectors, i.e.,
\begin{align*}
\colspace(\mat{Q}):= \clist{\mat{Q} \vect{x} | \vect{x} \in \R^{m}}.
\end{align*}
The orthogonal projection $\proj_{\colspace(\mat{Q})}(\vect{y})$ of a point $\vect{y} \in \R^{n}$  onto the column space $\colspace(\mat{Q})$ is given by (\reflem{lem.PointOrthogonalProjection}) 
\begin{align*}
\proj_{\colspace(\mat{Q})}(\vect{y}) = \argmin_{\vect{x} \in \colspace(\mat{Q})} \norm{\vect{x} - \vect{y}} = \mat{Q}\tr{\mat{Q}} \vect{y}.  
\end{align*}  
which is the orthogonal transformation of $\tr{\mat{Q}} \vect{y} \in \R^{m}$ via the linear map $\vect{x} \rightarrow \mat{Q} \vect{x}$.

The orthogonal projection $\proj_{\colspace{\mat{Q}}} \elp(\elpctr, \elpmat, \elprad)$ of an ellipsoid $\elp(\elpctr, \elpmat, \elprad)$ is another ellipsoid that is explicitly given by (\reflem{lem.EllipsoidOrthogonalProjection})
\begin{align*}
\proj_{\colspace(\mat{Q})} \elp(\elpctr, \elpmat, \elprad) &= \elp(\mat{Q} \tr{\mat{Q}} \elpctr, \mat{Q} \tr{\mat{Q}} \elpmat \mat{Q} \tr{\mat{Q}}, \elprad), 
\\
& = \clist{\mat{Q} \vect{x} | \vect{x} \in \elp(\tr{\mat{Q}}\elpctr, \tr{\mat{Q}} \elpmat \mat{Q}, \elprad)},
\end{align*}
which is the orthogonal transformation of the ellipsoid $\elp(\tr{\mat{Q}}\elpctr, \tr{\mat{Q}} \elpmat \mat{Q}, \elprad) \in \R^{m}$ via the linear map $\vect{x} \rightarrow \mat{Q} \vect{x}$.

Hence, the result follows from the definition of the metric projection since the system trajectory $\vect{y}(t)$ is contained in the Lyapunov ellipsoid $\elp(\mat{0}, \mat{P}^{-1}, \norm{\vect{y}(0)}_{\mat{P}})$.
\end{proof}


\bibliographystyle{IEEEtran}
\bibliography{references}


\end{document}

%% file: packages.tex
\usepackage{color}
\usepackage{comment}
\usepackage{verbatim}

\usepackage{cite}
\usepackage{url}

\usepackage{textcomp}
\usepackage{stfloats}
\usepackage{times}
\usepackage[mathcal]{euscript}
\usepackage{balance}

\usepackage[cmex10]{amsmath}
\usepackage{amssymb} 
\usepackage{amsfonts}

\usepackage{amsthm} 
  
\usepackage{dsfont} 

\usepackage{cuted}

\usepackage{graphicx}
\usepackage{epsfig} 
\usepackage[caption=false,font=normalsize,labelfont=sf,textfont=sf]{subfig}

\usepackage{array}
\usepackage{multirow}
\usepackage{enumerate}

\usepackage[ruled, vlined, linesnumbered]{algorithm2e}

%% file: commands.tex
\newtheoremstyle{plain}
	  {}
	  {}
	  {\itshape}
	  {}
	  {\bfseries}
	  {}
	  {5pt plus 1pt minus 1pt}
	  {}

\newtheoremstyle{definition}
  	  {}
	  {}
	  {\normalfont}
	  {}
	  {\bfseries}
	  {}
	  {5pt plus 1pt minus 1pt}
	  {}
  
\theoremstyle{plain}

\newtheorem{theorem}{Theorem}
\newtheorem{lemma}{Lemma}
\newtheorem{proposition}{Proposition}
\newtheorem{corollary}{Corollary}

\theoremstyle{definition}
\newtheorem{definition}{Definition}

\newcommand{\refeq}[1]			{(\ref{#1})} 
\newcommand{\reffig}[1]			{Fig. \ref{#1}} 
\newcommand{\refsec}[1]			{Section \ref{#1}}
\newcommand{\refapp}[1]			{Appendix \ref{#1}}
\newcommand{\reftab}[1]			{Table \ref{#1}}
\newcommand{\refthm}[1]			{Theorem \ref{#1}}
\newcommand{\refprop}[1]		{Proposition \ref{#1}}

\newcommand{\reflem}[1]			{Lemma \ref{#1}}

\newcommand{\reffn}[1] 		    {\textsuperscript{\ref{#1}}}

\newcommand{\R}  	{\mathbb{R}} 
\newcommand{\C}     {\mathbb{C}} 
\newcommand{\PDM} 	{S_{++}} 
\newcommand{\PSDM} 	{S_{+}} 

\let\originalleft\left
\let\originalright\right
\renewcommand{\left}{\mathopen{}\mathclose\bgroup\originalleft}
\renewcommand{\right}{\aftergroup\egroup\originalright}
	
\newcommand{\plist}[1] 	{\left(#1\right)} 
\newcommand{\blist}[1]	{\left[ #1 \right]} 
\newcommand{\clist}[1]	{\left\{#1\right\}} 

\newcommand{\vect}[1]   {\mathrm{#1}}
\newcommand{\mat}[1]    {\mathbf{#1}}
\newcommand{\tr}[1] {{#1}^{\mathrm{T}}} 
\newcommand{\norm}[1]  {\|#1\|}
\newcommand{\diag} {\mathrm{diag}} 

\newcommand{\argmin}{\operatornamewithlimits{arg\ min}} 
\newcommand{\diff} {\mathrm{d}}
\newcommand{\conv}      {\mathrm{conv}} 

%% file: notation.tex
\newcommand{\sposn} 	{\vect{x}} 
\newcommand{\sstate}	{\mat{x}} 
\newcommand{\sorder}    {n} 
\newcommand{\sdim}      {d} 
\newcommand{\sgain} 	{\kappa} 
\newcommand{\slambda}   {\lambda} 
\newcommand{\slambdaset} {\boldsymbol{\lambda}}
\newcommand{\scoef}     {\vect{c}} 

\newcommand{\cmat}       {\mat{C}} 
\newcommand{\vmat}       {\mat{V}} 
\newcommand{\vfunc}      {\nu} 
\newcommand{\vvect}      {\nu} 
\newcommand{\vsplx} {\mathcal{VS}} 
\newcommand{\esplx}	{\mathcal{ES}} 

\newcommand{\elp}		{\mathcal{E}} 
\newcommand{\elpctr}	{\vect{c}} 
\newcommand{\elpmat} 	{\mat{\Sigma}} 
\newcommand{\elprad}	{\rho} 

\newcommand{\proj}{\Pi} 
\newcommand{\colspace}{\mathcal{C}} 
